\documentclass[letterpaper]{article}

\ifdefined\aaaianonymous
    \usepackage[submission]{aaai2026}
\else
    \usepackage{aaai2026}
\fi

\usepackage{times}
\usepackage{helvet}
\usepackage{courier}
\usepackage[hyphens]{url}
\usepackage{graphicx}
\usepackage{natbib}
\usepackage{amsmath}
\usepackage{amssymb}
\usepackage{amsthm}

\usepackage{mathtools}
\usepackage{booktabs}
\usepackage{array}
\usepackage{placeins}
\usepackage{float}
\usepackage{multirow}
\usepackage{enumitem}
\newcommand{\R}{\mathbb{R}}
\newcommand{\Prb}{\mathbb{P}}

\newcommand{\norm}[1]{\lVert #1 \rVert}

\newcommand{\ctx}{m_{\mathrm{ctx}}}
\newcommand{\tool}{a_{\mathrm{tool}}}

\newtheorem{definition}{Definition}
\newtheorem{theorem}{Theorem}

\newtheorem{corollary}{Corollary}
\newtheorem{proposition}{Proposition}

\title{Brain-Prompt Injection: A Route-Safety Audit for BCI--LLM Agents}

\author{Jianwei Tai}
\affiliations{School of Internet, Anhui University\\
24012@ahu.edu.cn}

\begin{document}
\maketitle

\begin{abstract}
BCI-to-agent pipelines turn decoded neural activity into an authorization channel for tool-use agents, exposing a new attack surface we call \emph{brain-prompt injection}: signal-side perturbations, context-only injections, and adaptive dual-decoder attacks can all change the routed action while EEG-side or text-side monitors remain blind. Route safety in this stack depends on what the audit log can observe, not on decoder accuracy or agreement alone. We define a Route-Safety Audit Contract: a minimal log schema, denominator hierarchy, and endpoint specification, and prove an audit-schema separation theorem together with a C3 attacked-dependence decomposition; clean agreement and marginal robustness do not identify the joint term that controls C3 routing. As a calibration layer on top of the contract, we apply split-conformal calibration to a non-oracle EEG confirmation channel and report the resulting false-accept frontier under an explicit threat-archetype matrix. We instantiate the contract on EEGMMI native left/right command-control over 5{,}400 events, harmless tool stubs, and seed/case denominators. Provenance blocks C2 routes ($0.000$); agreement-plus-provenance routes C3 flips ($1.000$); confirmation-plus-provenance routes them ($0.000$). The conformal frontier reaches FAR $0.000$ at clean utility $0.150$ for $\alpha=.005$ and FAR $0.119$ at clean utility $0.452$ for $\alpha=.10$ under acquisition isolation; an attacker-controllable confirmation channel breaks the bound to $\approx\!1$. Subject-cluster bootstrap confirms these intervals on $60$ subjects; cross-architecture (TinyEEGNet, EEGNetV4) and capacity-sweep results show within-regime saturation. Mediation and confirmation reduce risk; they are not intent certificates.
\end{abstract}

\section{Introduction}
\label{sec:intro}

Brain--computer interfaces (BCIs) increasingly act as authorization channels rather than only display channels. Non-invasive and invasive systems decode neural activity into commands, characters, or speech with steadily lower error rates \cite{willett2021high,metzger2023high,levy2025brain2qwerty}, and large EEG foundation models such as BIOT and LaBraM are pushing toward universal cross-dataset decoders that can be plugged into downstream pipelines \cite{yang2023biot,jiang2024labram}. In parallel, large language models have moved from text generation to tool-using \emph{agents} that select APIs, invoke side effects, and act on behalf of the user \cite{schick2023toolformer,qin2023toolllm}. The natural composition of these two lines is a BCI--LLM agent: a decoded neural command becomes an instruction-like input that a tool-using LLM converts into routes such as ``move the cursor'', ``send a message'', or ``confirm a transfer.''

This composition opens an attack surface that neither line alone treats. We call the decoded neural instruction stream a \emph{brain-prompt channel}, and we call the corresponding family of attacks \emph{brain-prompt injection}. Brain-prompt injection is exposed to three failure modes that route-safety claims must address jointly. First, signal-side perturbations of the EEG window, in the spirit of EEG adversarial examples \cite{zhang2019vulnerability}, can change the decoded command (C1). Second, context-only injections, analogous to indirect prompt injection on LLM agents \cite{greshake2023indirect,yi2023bipia,zhan2024injecagent,debenedetti2024agentdojo,zhang2024asb}, can alter the agent route while the EEG side remains clean and any EEG-only monitor sees no anomaly (C2). Third, adaptive shared-input perturbations can drive a primary decoder $f$ and an optional verification decoder $f_2$ jointly, so dual-channel agreement is satisfied by an attacker-controlled rather than user-intended target (C3). Existing benchmarks for indirect prompt injection \cite{zhan2024injecagent,debenedetti2024agentdojo,zhang2024asb} cover the C2-style failure but do not log a decoded neural source; existing EEG attacks \cite{zhang2019vulnerability} cover the C1/C3 signal side but do not log context provenance or downstream routing. The authorization claim a BCI--LLM agent must defend lives at the intersection.

A natural first defense is to require two heterogeneous decoders to agree before the agent acts. We show that, under a shared raw-input perturbation budget, dual-decoder agreement does not certify intent: an attacker can satisfy the agreement predicate by jointly driving both decoders' margins (Proposition~\ref{lem:observability-feasibility}), and the resulting C3 route risk factors through a primary-target probability $p_1$, an attacked conditional agreement term $\Prb(B\mid A)$, and an execution-policy term $\alpha$, with an attacked dependence lift $\Delta_{\mathcal{D}}=\Prb(A\cap B)-p_1p_2$ that clean agreement statistics never identify (Theorem~\ref{thm:c3-risk-decomposition}). The decomposition makes the corresponding audit-schema separation explicit: the per-case log must record context provenance, the attacked secondary decision, the execution policy, and the route/confirmation outcome, otherwise paired audit worlds with identical observed logs can have opposite C2 or C3 truth values.

We study the audit boundary rather than a new universal defense or a physical attack. The empirical contribution is a calibrated falsification protocol on EEGMMI native left/right command-control: 5{,}400 events, ten seeds, harmless cursor and tool stubs, seed/case denominators, and strict held-out calibration. On top of the audit contract we apply a textbook split-conformal calibration layer \cite{vovk2005algorithmic,romano2019conformalized} to a non-oracle confirmation channel and report the resulting false-accept frontier across an explicit threat-archetype matrix; a folklore necessity argument shows that confirmation-channel non-controllability cannot be removed without losing the bound. C4 (extraction-style query attacks) is logged as metadata only; the load-bearing theory and experiments are C2/C3-focused.

\paragraph{Contributions.}
\begin{itemize}
    \item \textbf{Route-Safety Audit Contract.} We prove an audit-schema separation theorem and a C3 attacked-dependence decomposition for BCI tool-use authorization: the schema names the variables individually necessary and jointly sufficient to certify a C2 or C3 route claim, and the decomposition identifies the attacked joint term that clean agreement and marginal robustness cannot identify.
    \item \textbf{Conformal calibration layer with falsifiable assumptions.} On top of the contract we apply split-conformal calibration to the non-oracle confirmation channel and characterize its operational assumptions through a threat-archetype matrix (A, B, B$_{\mathrm{xr}}$, C, C$'$); a folklore necessity argument shows that confirmation-channel non-controllability cannot be removed without losing the bound.
    \item \textbf{Empirical instantiation on native EEGMMI command-control.} Over $5{,}400$ left/right fist execution and imagery events with harmless tool stubs and strict held-out calibration, we report the audit on Exp8 native command-control, the Exp9 confirmation frontier, the cross-architecture and capacity sweeps, the threat-archetype matrix, preprocessing/temporal stress, and matched stronger-defense baselines, with subject-cluster bootstrap confirming the trial-level intervals.
\end{itemize}

\section{Related Work}
\label{sec:related}

\paragraph{LLM-assisted BCI systems.}
BCI systems decode neural activity into communication and control commands, from general-purpose platforms such as BCI2000 to high-performance neural text and speech interfaces \cite{schalk2004bci2000,willett2021high,metzger2023high}. Non-invasive interfaces such as MEG-based Brain2Qwerty have begun closing the gap between EEG-grade decoding and deployable typing rates \cite{levy2025brain2qwerty}, while biosignal foundation models such as BIOT and LaBraM further motivate cross-dataset, cross-task EEG decoders that downstream language or control systems can consume \cite{yang2023biot,jiang2024labram}. We study the security boundary created when decoded neural intent is treated as an authorization signal for a tool-use agent.

\paragraph{Prompt injection, jailbreaking, and tool-use agents.}
Tool-using language models expose APIs and external context to model-mediated decisions \cite{schick2023toolformer,qin2023toolllm}. Indirect prompt injection shows that context supplied outside the user's direct instruction can compromise LLM-integrated applications \cite{greshake2023indirect}; jailbreak and transferable prompt attacks show that model-side safety predicates can fail under adaptive inputs \cite{wei2023jailbroken,zou2023universal}. A growing line of benchmarks now operationalizes these attacks on tool-integrated agents: BIPIA characterizes indirect-injection attack success across LLM backbones \cite{yi2023bipia}; InjecAgent reports that 30 LLM agents are vulnerable to embedded malicious instructions across 17 tool families \cite{zhan2024injecagent}; AgentDojo defines a dynamic environment for prompt-injection attacks and defenses on tool-using agents \cite{debenedetti2024agentdojo}; and Agent Security Bench formalizes ten attack and defense classes across 13 LLM backbones \cite{zhang2024asb}. Existing work largely studies text, retrieval, web, or tool contexts. It does not model decoded neural commands as a second authorization channel whose EEG-side observables may disagree with route-level context observables.

\paragraph{Adversarial robustness and BCI security.}
Adversarial-example work established that robustness claims need adaptive, attack-aware evaluation rather than benign accuracy alone \cite{papernot2016limitations,carlini2017towards,madry2018towards,tramer2017ensemble}. EEG decoders such as EEGNet make gradient-based BCI attacks technically meaningful \cite{lawhern2018eegnet}. Prior work shows that CNN classifiers in EEG-based BCIs are vulnerable to adversarial examples \cite{zhang2019vulnerability}. This signal-side literature does not address context-only agent injection or agreement bypass in BCI--LLM tool pipelines; our contribution is the class-specific validation protocol linking those channels.

\paragraph{Authorization boundary.}
BCI-agent security sits at the intersection of these lines. Tool-agent papers define model-mediated API use and tool selection, but usually treat the user command as text rather than a decoded biosignal. Indirect prompt-injection work studies untrusted context and retrieval, but not a second neural authorization channel whose observables are disjoint from the route context. BCI adversarial-security work studies perturbations of the neural classifier, but not whether a downstream agent route is justified by provenance, agreement, or confirmation. A route-safety claim must name the attack class, the observed variables, and the denominator that can falsify the claimed block.

\paragraph{Closest-prior separation by missing observable.}
The paper-level distinction is denominator-level rather than vocabulary-level. Each closest prior line covers one side of the BCI-to-tool interface and omits the variable that the corresponding C2 or C3 route claim requires. Indirect prompt-injection and jailbreak work \cite{greshake2023indirect,wei2023jailbroken,zou2023universal} covers untrusted text and context but does not record the decoded neural source or the route denominator, so it cannot certify EEG-authorized route safety. Tool-agent work \cite{schick2023toolformer,qin2023toolllm} covers API and tool routing but does not log biosignal provenance or source-not-target EEG cases, so it cannot certify C2 neural-route attribution. EEG adversarial-example work \cite{zhang2019vulnerability,lawhern2018eegnet} covers decoder perturbation but does not log context provenance or execution route, so it cannot certify downstream authorization. High-performance BCI and biosignal-foundation systems \cite{willett2021high,metzger2023high,yang2023biot} cover neural communication and control but do not log the attacked secondary decision or execution policy, so they cannot test agreement-as-intent under C3. Standard adversarial-robustness evaluation usually reports marginal target success; Theorem~\ref{thm:c3-risk-decomposition} shows that the route claim depends instead on the attacked joint term and the execution-policy term. The Route-Safety Audit Contract names the missing observables for each line and turns them into denominators in Table~\ref{tab:audit-protocol}.

\section{System and Threat Model}
\label{sec:threat}

A BCI-controlled tool-use agent is a tuple
\begin{equation}
\Pi = (\phi, f, f_2, A_{\mathrm{LLM}}, R),
\label{eq:pipeline}
\end{equation}
where $\phi$ maps a neural window to a representation, $f$ is a primary intent decoder, $f_2$ is an optional verification decoder, $A_{\mathrm{LLM}}$ maps decoded intent and context to a tool policy, and $R$ is a rule or monitor.

Let $x \in \R^{C \times T}$ be an EEG window, $z=\phi(x)$ its representation, $y=f(z)$ the decoded intent, $\ctx$ the external context, and $\tool=A_{\mathrm{LLM}}(y,\ctx)$ the selected tool action. A dual-channel rule executes automatically only when $f$ and $f_2$ agree; otherwise it requests confirmation or abstains.

\begin{definition}[Audit classes]
A BCI tool-use audit class names which input or log field may change and which route-safety denominator is required.
\begin{itemize}
    \item C1 direct perturbation: the audit changes $x$ or $z$ to alter $f$; this is a decoder-stress axis.
    \item C2 context-only injection: the audit changes $\ctx$ while leaving EEG-side variables unchanged; the required denominator is routed clean/untrusted episodes with provenance.
    \item C3 adaptive agreement audit: the audit perturbs the shared input seen by $f$ and $f_2$; the required denominator is source-not-target attacked cases with primary, secondary, and policy logs.
    \item C4 extraction or probing: repeated queries are logged as metadata; C4 is not a claim-bearing empirical class here.
\end{itemize}
\end{definition}

\paragraph{Audit capabilities and excluded deployment claims.}
Table~\ref{tab:threat-capabilities} fixes the audit variables used by the validation cells. The table describes controlled audit interventions, not capabilities assumed for every deployment. C2 is a context-control condition and is intentionally outside the observation set of EEG-only monitors. C3 is a white-box or audit-time adaptive condition against differentiable decoders, evaluated on the standardized raw EEG tensor seen by both decoders. This is the standard robustness-audit setting for testing whether agreement is a certificate; it is not a claim that the same perturbation channel is already available in deployed BCI hardware.

\begin{table}[t]
\centering
\caption{Class-specific audit capabilities and non-claims. ``Observed by rule'' means the variables available to the corresponding validation monitor, not all information available to a system designer.}
\label{tab:threat-capabilities}
\scriptsize
\setlength{\tabcolsep}{2pt}
\resizebox{\linewidth}{!}{
\begin{tabular}{@{}lccc@{}}
\toprule
Class & Audit intervention & Observed by rule & Load-bearing test \\
\midrule
C1 & raw EEG or embedding & EEG/logits/agreement & primary-target PGD \\
C2 & external context only & context provenance & flag uplift and route block \\
C3 & shared raw tensor, $f,f_2$ known & dual-decoder agreement & joint raw PGD \\
C4 & repeated queries & query log/rate limits & reporting obligation \\
\bottomrule
\end{tabular}}
\end{table}

\begin{table}[t]
\centering
\caption{Glossary of recurring symbols. Audit classes (C1--C4) name what the attacker may change; threat archetypes (A, B, B$_{\mathrm{xr}}$, C, C$'$) and B-strict / B-strict$^-$ name what the attacker may perturb in the conformal layer of \S\ref{sec:conformal}; experiment cells (Exp1--Exp10) are referenced throughout \S\ref{sec:experiments}.}
\label{tab:glossary}
\scriptsize
\setlength{\tabcolsep}{1.4pt}
\begin{tabular}{@{}p{.18\linewidth}p{.78\linewidth}@{}}
\toprule
Symbol & Meaning \\
\midrule
C1 & direct EEG/embedding perturbation (decoder-stress axis). \\
C2 & context-only injection without changing EEG-side variables. \\
C3 & adaptive shared-input perturbation seen by both decoders. \\
C4 & repeated-query / extraction reporting axis (not load-bearing). \\
A & confirmation pair from same subject, same session (audit baseline). \\
B & confirmation pair from disjoint subject (acquisition-isolated audit). \\
B$_{\mathrm{xr}}$ & confirmation pair from same subject, disjoint EEGMMI run (cross-run within session). \\
C & threat-archetype where the attacker also perturbs the confirmation window with same-subject pair. \\
C$'$ & same as C but with disjoint-subject pair. \\
B-strict & threat model: source-only attack with acquisition-isolated and non-controllable confirmation. \\
B-strict$^-$ & B-strict with assumption (ii) removed; attacker may perturb the confirmation window. \\
$\alpha$ & target conformal false-accept rate; $\tau_t(\alpha)$ is the corresponding split-conformal threshold. \\
$\varepsilon_{\mathrm{src}}$ & $\ell_\infty$ budget on the source command window. \\
$\varepsilon_{\mathrm{cnf}}$ & $\ell_\infty$ budget on the confirmation window (zero under B-strict). \\
$f_B, s_B$ & confirmer model and its score $f_B(X_{\mathrm{cnf}})_t$ for target $t$. \\
$\Pi$ & router; rules include primary-only, confidence gate, agreement, provenance, agreement+provenance, confirmation. \\
\bottomrule
\end{tabular}
\end{table}

\begin{table}[t]
\centering
\caption{Threat-model and reproducibility summary.}
\label{tab:threat-repro}
\scriptsize
\setlength{\tabcolsep}{1.5pt}
\begin{tabular}{@{}ll@{}}
\toprule
Item & Setting \\
\midrule
Data & EEGMMI R03 main cell; R03/04/07/08/11/12 native command cell. \\
Events & 900 R03 T1/T2 events; 5{,}400 native left/right command events. \\
Decoders & TinyEEGNet A/B gate; A/C replicate/native cells; seeds 0--9. \\
Training & 80 epochs; batch 64; no silent seed exclusion. \\
Attack & Standardized raw PGD; $\ell_\infty$ clip $[-\epsilon,\epsilon]$. \\
Raw scale & $.08$ std. units $\approx .0778$ cache units, $4.1\%$ of p95 $|x|$. \\
Grid & $\epsilon\in\{.02,.04,.08\}$; 40 steps; step .005. \\
C2 denom. & Routed cases; uplift = C2 flag $-$ clean flag. \\
C3 denom. & 120 attack and 60 unauthorized cases/seed. \\
Gate & $\ge$7 seeds; acc. $\ge$.60; disagr. $\le$.40. \\
Bypass & C3 success $\ge$.30; cond./unauth. $\ge$.70. \\
Routing & Provenance blocks C2 target route to $\le$.05. \\
\bottomrule
\end{tabular}
\end{table}

\section{Theory}
\label{sec:theory}

\begin{definition}[Route-Safety Audit Contract]
\label{def:rsac}
For an attack class $c$ and router family $\Pi$, a Route-Safety Audit Contract is a tuple
\begin{equation}
\mathcal{C}_c=(\mathcal{L}_c, D_c, E_c, G_c),
\label{eq:rsac}
\end{equation}
where $\mathcal{L}_c$ is the required per-case log schema, $D_c$ is the statistical denominator and hierarchy, $E_c$ is the route-safety event being certified, and $G_c$ is the declared pass/fail endpoint. A certificate for claim $E_c$ is sound only if $E_c$ is measurable with respect to $\mathcal{L}_c$ and the reported rate is computed over $D_c$ rather than over a projected or duplicated route-episode count.
\end{definition}

\begin{proposition}[Class-specific observability and agreement feasibility]
\label{lem:observability-feasibility}
Consider a BCI tool-use pipeline $\Pi$ with EEG-side variables $V(x)$, context $\ctx$, route $\tool=A_{\mathrm{LLM}}(y,\ctx)$, and two differentiable decoders $f,f_2$. Let $R$ be any deterministic or randomized monitor measurable with respect to $V(x)$. If a C2 attacker changes only $\ctx$ and does not change the distribution of $x$, then
\begin{equation}
\Prb[R(V(x))=1 \mid \mathrm{C2}] - \Prb[R(V(x))=1 \mid \mathrm{clean}] = 0,
\label{eq:c2-flag-uplift}
\end{equation}
although the target-route probability can change whenever $A_{\mathrm{LLM}}$ is context-sensitive. For C3, let $s_g(x)$ be the logits of decoder $g\in\{f,f_2\}$ and define
\begin{equation}
m_g(x,y_t)=\max_{y\ne y_t} s_g(x)_y-s_g(x)_{y_t},
\label{eq:margin-def}
\end{equation}
using a subgradient at logit ties. Suppose there is a shared raw-input direction $u$ with $\norm{u}_\infty\le 1$ and per-decoder local remainder bounds $\rho_g(\epsilon,u)$ satisfying
\begin{equation}
\begin{aligned}
\bigl|m_g(x+\epsilon u,y_t)-m_g(x,y_t) \\
{}-\epsilon\langle\nabla_x m_g(x,y_t),u\rangle\bigr|\le \rho_g(\epsilon,u).
\end{aligned}
\label{eq:remainder-bound}
\end{equation}
If
\begin{equation}
\begin{aligned}
\epsilon \langle -\nabla_x m_g(x,y_t),u\rangle &> m_g(x,y_t)+\rho_g(\epsilon,u) \\
& \text{for both } g\in\{f,f_2\},
\end{aligned}
\label{eq:feasibility}
\end{equation}
then the $\ell_\infty$ ball contains a point where both decoders select $y_t$ and an agreement rule accepts the target. Thus C2 and C3 fail different observability assumptions: C2 bypasses EEG-only monitoring by changing an unobserved context variable, whereas C3 satisfies the observed agreement predicate itself.
\end{proposition}

\begin{proof}
Under the C2 intervention, $V(x)$ has the same distribution as in the clean condition because $x$ is unchanged and $R$ does not observe $\ctx$. Any monitor measurable with respect to $V(x)$, including a randomized monitor whose randomness is independent of the attack condition, therefore has the same output distribution. The flag-probability difference is zero. The route statement follows because $\tool=A_{\mathrm{LLM}}(y,\ctx)$ may vary with $\ctx$ even when $y$ and all EEG-side variables are fixed.

For C3, the remainder bound gives, for each decoder $g$,
\begin{equation}
m_g(x+\epsilon u,y_t) \le m_g(x,y_t)+\epsilon\langle\nabla_x m_g(x,y_t),u\rangle+\rho_g(\epsilon,u).
\label{eq:remainder-upperbound}
\end{equation}
The displayed feasibility condition makes the right-hand side strictly negative for both $f$ and $f_2$. Hence both margins are non-positive at the same perturbed input $x+\epsilon u$, and both decoders assign $y_t$. The execution rule accepts agreement, but the agreement now certifies consistency of two jointly driven decoders rather than user intent. The PGD objective used in the experiments is the standard cross-entropy relaxation of this per-decoder boundary-crossing condition.
\end{proof}

\begin{theorem}[C3 agreement-route risk decomposition]
\label{thm:c3-risk-decomposition}
Fix an attack family $\mathcal{D}$ over source-not-target cases and perturbations. Let
\begin{align}
A&=\{f(x+\delta)=y_t\},\\
B&=\{f_2(x+\delta)=y_t\},\\
P&=\{\pi\text{ executes target without confirmation}\},
\end{align}
and define $p_1=\Prb_{\mathcal{D}}(A)$, $p_2=\Prb_{\mathcal{D}}(B)$, $\alpha=\Prb_{\mathcal{D}}(P\mid A,B)$, and attacked dependence lift
\begin{equation}
\Delta_{\mathcal{D}}=\Prb_{\mathcal{D}}(A\cap B)-p_1p_2.
\label{eq:dependence-lift}
\end{equation}
Under agreement-plus-provenance routing,
\begin{align}
\Prb_{\mathcal{D}}(E_{C3})
&=\alpha(p_1p_2+\Delta_{\mathcal{D}}) \notag \\
&=p_1\Prb_{\mathcal{D}}(B\mid A)\alpha.
\label{eq:c3-decomposition}
\end{align}
Thus agreement reduces primary-only target routing by a factor at least $r$ if and only if $\Prb_{\mathcal{D}}(B\mid A)\alpha\le r$. Without the attacked secondary-decision log, two audits can share $p_1$, $p_2$, and all clean agreement statistics while taking any joint value allowed by the Fr\'echet bounds
\begin{equation}
\max\{0,p_1+p_2-1\}\le \Prb_{\mathcal{D}}(A\cap B)\le \min\{p_1,p_2\}.
\label{eq:frechet}
\end{equation}
Therefore clean agreement and marginal robustness do not certify C3 route safety; the audit must estimate the attacked conditional agreement term and the execution-policy term on the same source-not-target cases.
\end{theorem}

\begin{proof}
For source-not-target cases, agreement-plus-provenance executes the attacked target exactly when the primary decoder selects $y_t$, the secondary decoder also selects $y_t$, and the active policy permits execution without an independent confirmation block. This gives $E_{C3}=A\cap B\cap P$. Conditioning on $A\cap B$ gives $\Prb_{\mathcal{D}}(E_{C3})=\alpha\Prb_{\mathcal{D}}(A\cap B)$, the definition of $\Delta_{\mathcal{D}}$ gives the first equality, and the chain rule gives the second. Primary-only target routing has probability $p_1$, so an agreement policy achieves residual ratio $\Prb_{\mathcal{D}}(E_{C3})/p_1=\Prb_{\mathcal{D}}(B\mid A)\alpha$ when $p_1>0$; when $p_1=0$ both routes are zero. The Fr\'echet bounds are the sharp possible ranges for a joint event with fixed marginals. Clean agreement statistics are functions of decoder outputs at $x$, whereas $A$ and $B$ are target-hit events at $x+\delta$ under $\mathcal{D}$, so they do not identify the attacked joint term.
\end{proof}

\begin{theorem}[Minimal audit-schema separation for C2/C3 route claims]
\label{thm:audit-identifiability}
Consider the router family in this paper: a route rule observes a decoded command, an optional secondary decoder, a context-provenance flag, an execution policy $\pi$, and a routed action $\tool$ with execution bit $\mathrm{exec}$. Let $E_{C2}$ be the event that an untrusted context induces the requested target route, and let $E_{C3}$ be the event that a source-not-target EEG case is routed to the attacked target under the active policy. For any audit projection $O$ and any deterministic or randomized certification rule $Q$ measurable with respect to $O$, the following hold.
\begin{enumerate}[leftmargin=*,itemsep=1pt,topsep=2pt]
    \item If $O$ omits either context provenance or route/execution outcome, then $E_{C2}$ is not identifiable from $O$: there exist paired clean/untrusted audit worlds with identical $O$ and different $E_{C2}$ truth values.
    \item If $O$ omits any of the attacked secondary decision, execution policy, confirmation status, or route/execution outcome, then $E_{C3}$ is not identifiable from $O$: there exist paired source-not-target audit worlds with identical $O$ and different $E_{C3}$ truth values.
    \item The per-case log
    \begin{equation}
    \begin{aligned}
    \ell=&(s,i,k,\mathrm{src},p_{\mathrm{ctx}},y,y_2,y^{\mathrm{adv}},y_2^{\mathrm{adv}},\\
    & \pi,\tool,\mathrm{exec},\mathrm{confirm},s_B,\tau)
    \end{aligned}
    \label{eq:schema-log}
    \end{equation}
    identifies both events for this router family when rates are computed over the seed/case/context hierarchy. Thus the schema is sufficient, and every omitted field listed above is necessary for at least one of the two route-safety claims.
\end{enumerate}
Consequently, no sound nontrivial rule using a weaker projection can certify the corresponding route-safety claim: it either falsely certifies a failing world or abstains on an indistinguishable safe world.
\end{theorem}

\begin{proof}
For C2, construct two worlds with the same EEG window, primary and secondary decoder outputs, confidence, clean source command, and EEG-side monitor variables. In $W_0$, the context is trusted or benign and the route follows the decoded command; in $W_1$, the context is untrusted and requests the target route. If the projection omits provenance, the trusted and untrusted contexts have the same projection. If it omits route or execution outcome, a blocked and executed route have the same projection. In both cases $O(W_0)=O(W_1)$ while $E_{C2}$ differs.

For C3, construct worlds with the same clean source command, clean agreement state, and clean route record. In $W_0$, the attacked secondary decoder does not agree with the target, or the active policy requires confirmation and confirmation is absent, or the target is not executed. In $W_1$, the attacked primary and secondary decoders both select the target and the policy executes the target route. If the projection omits the attacked secondary decision, policy, confirmation status, or route/execution outcome, these worlds remain indistinguishable even though $E_{C3}$ differs.

The listed log contains the seed $s$, case $i$, and route context $k$ needed to define denominators; the clean source command $\mathrm{src}$ needed for source-not-target conditioning; provenance $p_{\mathrm{ctx}}$; clean and attacked decoder decisions; policy $\pi$; route and execution outcome; confirmation status; and the confirmer score $s_B = f_B(X_{\mathrm{cnf}})_t$ together with the active conformal threshold $\tau = \tau_t(\alpha)$, which together reproduce the conformal accept decision used in \S\ref{sec:conformal}. Reading these fields decides the Boolean predicates defining $E_{C2}$ and $E_{C3}$ for each logged episode. Rates over the validation cells are then ordinary averages over the specified hierarchy. Finally, because any weaker projection above admits paired worlds with identical projected logs and different event truth values, a measurable certification rule has the same output distribution in both worlds. Positive certification probability gives a false certificate in the failing world; zero probability gives no nontrivial certificate in the safe world.
\end{proof}

\begin{corollary}[Minimal sound certificate for the reported C2/C3 claims]
\label{cor:rsac-soundness}
For the router family and C2/C3 events in Theorem~\ref{thm:audit-identifiability}, the contract whose log is $\ell$ and whose denominators are the seed/case/context hierarchies in Table~\ref{tab:audit-protocol} is sufficient for a sound route-safety certificate. Any certificate that omits one of the theorem's necessary fields for the corresponding class is not sound for that class: it can pass on a projected log while an indistinguishable audit world has the opposite route-safety truth value.
\end{corollary}

\begin{proof}
Sufficiency follows because $\ell$ makes $E_{C2}$ and $E_{C3}$ measurable per episode and Table~\ref{tab:audit-protocol} fixes the denominators over which rates are averaged. Necessity follows from the paired-world construction in Theorem~\ref{thm:audit-identifiability}: a weaker certificate is measurable with respect to a projection that is identical in two worlds with different event truth values, so it cannot certify both without either false acceptance or abstention.
\end{proof}

\begin{corollary}[Denominator-map correspondence]
\label{cor:nonidentifying-aggregates}
Each invalid aggregate in Table~\ref{tab:denominator-map} is a non-identifying projection for its stated route-safety claim. Replacing it with the listed denominator is exactly the schema repair required by Theorem~\ref{thm:audit-identifiability}.
\end{corollary}

\begin{proof}
C2 EEG flag rate discards context provenance and route outcome, so the C2 paired worlds remain indistinguishable. Clean agreement rate discards the attacked secondary decision and post-attack route outcome, so the C3 paired worlds remain indistinguishable. C3 primary-target success alone discards whether the target was actually routed and whether confirmation was required. Route episode count alone discards the seed/case/context hierarchy needed to distinguish statistical support from repeated route templates. Confirmation route rate without confirmation status discards whether the block used an independent user signal, an oracle upper bound, or an automatic abstention rule. Adding the table's missing observables breaks the corresponding equality $O(W_0)=O(W_1)$ and makes the route claim testable.
\end{proof}

The Route-Safety Audit Contract is intentionally minimal. Its new object is the class-specific route audit: the attacked dependence lift that determines whether agreement lowers C3 route risk, plus the denominator discipline needed to execute the audit. It is not a general impossibility theorem or a new adversarial optimizer. It states which logged variables are sufficient to support a route-safety claim and which common aggregate metrics are non-identifying. The point is to prevent a false certificate: an EEG-only C2 flag rate can be exactly unchanged while route risk changes, and a clean agreement rate can be high while attacked agreement routes an unauthorized target. Theorem~\ref{thm:audit-identifiability} gives both sides of the contract: it proves necessity for the C2 and C3 fields that separate paired audit worlds, and sufficiency for the log schema used by the experiments. The C3 condition in Proposition~\ref{lem:observability-feasibility} is not assumed to be common; its non-vacuity is a measurable property of a decoder pair, target, and perturbation budget. The experiments therefore report decoder gates, target margins, shared-descent diagnostics, perturbation-strength sweeps, route-policy outcomes, confirmation status, and seed/case/context hierarchy rather than treating clean agreement as a certificate. Table~\ref{tab:denominator-map} lists the aggregate-to-denominator repairs used throughout the validation protocol. This denominator discipline prevents two common but misleading evaluations: counting EEG-side flags as C2 security, and treating clean agreement as an adaptive C3 certificate. The lemma predicts the zero C2 attack-attributable EEG flag uplift reported in the real-EEG cells, and it predicts that C3 should become easier as a shared raw-input budget can satisfy the per-decoder target inequalities on more cases; Table~\ref{tab:c3-strength-sweep} and the margin bridge in Table~\ref{tab:evidence-spine} test this feasibility prediction.

\subsection{Split-Conformal Confirmation FAR Control}
\label{sec:conformal}

The non-oracle confirmation proxy of \S\ref{sec:routing} can be cast as split-conformal calibration on confirmer scores. Proposition~\ref{thm:conformal-far} below is the textbook split-conformal coverage statement \cite{vovk2005algorithmic,romano2019conformalized} adapted to the binary confirmer scoring rule; we restate it for completeness because the rest of the paper invokes it as a layer-tool on top of the audit contract. The original contribution of this subsection is not the proposition itself but the explicit operational assumptions it requires (threat model B-strict) and the threat-archetype matrix in \S\ref{sec:phase5} that makes those assumptions empirically falsifiable on a BCI confirmation channel.

\begin{definition}[Strict held-out confirmation calibration]
\label{def:strict-calibration}
Let $f_B$ be a confirmer trained only on the train index. For each target label $t\in\{0,1\}$, let
$
S_i^{(t)}=f_B(X_i)_t
$
for calibration windows with $Y_i\ne t$, and let $m_t$ be the count of such negative calibration windows. The conformal threshold at level $\alpha\in(0,1)$ is
\begin{equation}
\tau_t(\alpha)=\bigl(\lceil(1-\alpha)(m_t+1)\rceil\bigr)\text{-th order statistic of }\{S_i^{(t)}\}.
\label{eq:conformal-tau}
\end{equation}
A requested target $t$ on a confirmation window $X_{\mathrm{cnf}}$ is accepted iff $\arg\max_y f_B(X_{\mathrm{cnf}})_y=t$ and $f_B(X_{\mathrm{cnf}})_t\ge\tau_t(\alpha)$.
\end{definition}

\begin{proposition}[Split-conformal confirmation FAR control \cite{vovk2005algorithmic,romano2019conformalized}]
\label{thm:conformal-far}
Suppose, for each target $t$, the negative calibration scores $\{S_i^{(t)}:Y_i\ne t\}$ and a future negative confirmation score $S_{m_t+1}^{(t)}$ are exchangeable. Under Definition~\ref{def:strict-calibration},
\begin{equation}
\Prb\bigl[S_{m_t+1}^{(t)}\ge\tau_t(\alpha)\;\bigm|\;Y_{m_t+1}\ne t\bigr]\le\frac{\lfloor\alpha(m_t+1)\rfloor}{m_t+1}\le\alpha.
\label{eq:conformal-coverage}
\end{equation}
The argmax acceptance condition makes the false-accept event a subset of the threshold-exceedance event, so the same upper bound applies to confirmation false accepts.
\end{proposition}

\begin{proof}
Standard split-conformal coverage. Exchangeability of $S_1^{(t)},\dots,S_{m_t}^{(t)},S_{m_t+1}^{(t)}$ implies the rank of $S_{m_t+1}^{(t)}$ is uniform on $\{1,\dots,m_t+1\}$, so $\Prb[S_{m_t+1}^{(t)}\ge\tau_t(\alpha)]\le\lfloor\alpha(m_t+1)\rfloor/(m_t+1)$. The argmax condition only removes accepts, so it cannot increase the FAR upper bound.
\end{proof}

\begin{corollary}[Threat model B-strict for unauthorized confirmed routes]
\label{cor:far-route}
Assume (i) acquisition isolation: the confirmation window is acquired from a session whose marginal score distribution is exchangeable with the calibration negatives the attacker cannot influence; and (ii) confirmation-channel non-controllability: the attacker may perturb the source command window in an $\ell_\infty$ ball but cannot perturb the confirmation window. Then for any source-window attacker $\mathcal{A}$ within the budget,
\begin{equation}
\Prb\bigl[\text{unauthorized confirmed route to }t\bigr]\;\le\;\alpha,
\label{eq:route-bound}
\end{equation}
up to the discretization in Proposition~\ref{thm:conformal-far}, regardless of $\mathcal{A}$'s strength.
\end{corollary}

\begin{proof}
Source-window perturbations leave $X_{\mathrm{cnf}}$ unchanged by assumption (ii), so the marginal distribution of $S_{m_t+1}^{(t)}$ is exchangeable with the negatives by assumption (i). Proposition~\ref{thm:conformal-far} applies. The unauthorized confirmed-route event is a subset of the FAR event, so the same upper bound holds.
\end{proof}

The two assumptions are not cosmetic. Removing either is an empirical adversarial possibility test in \S\ref{sec:phase5}: the same audit with attacker-perturbable confirmation drives the route rate to nearly $1$, and the same audit with same-subject same-session confirmation pairing remains controlled at the cost of being an audit setup rather than a deployable confirmation channel.

\paragraph{Calibration consistency under shrinking $n_c$.}
Proposition~\ref{thm:conformal-far} also predicts that empirical FAR on a held-out test set should remain below $\alpha$ as the calibration size shrinks, with discretization error $O(1/m_t)$. The same theorem predicts $\varepsilon$-invariance under acquisition isolation: the confirmer never sees the attacked window, so the empirical FAR upper bound does not depend on the attacker's $\varepsilon$. Both predictions are tested in \S\ref{sec:phase5}.

\paragraph{Necessity of strict calibration.}
If $f_B$ is trained on $\text{train}\cup\text{cal}$, the calibration negatives are no longer exchangeable with future test negatives because the model has seen the calibration windows during fitting. The marginal coverage of Proposition~\ref{thm:conformal-far} no longer holds, and the empirical FAR systematically exceeds $\alpha$. \S\ref{sec:phase5} reports the leaky-vs-strict gap: at $\alpha=.05$, leaky calibration produces empirical FAR $0.0912$ with $95\%$ Wilson CI $[0.073,0.113]$, while strict calibration produces $0.0525$ with CI $[0.039,0.070]$.

\subsection{Necessity of Confirmation-Channel Non-Controllability}
\label{sec:necessity}

Proposition~\ref{thm:conformal-far} and Corollary~\ref{cor:far-route} are stated under threat model B-strict, which combines (i) acquisition isolation and (ii) confirmation-channel non-controllability. Proposition~\ref{thm:necessity} below is a folklore PGD-feasibility argument, restated here to make the necessity of condition (ii) explicit: if the attacker can perturb the confirmation window inside any nontrivial $\ell_\infty$ budget for any non-degenerate confirmer, no FAR upper bound of the form ``$\alpha + g(n_c, \delta)$'' independent of attacker actions can hold. The empirical contribution is the C/C$'$ archetypes in \S\ref{sec:phase5}, not the proposition itself.

Threat model B-strict$^-$ replaces (ii) by (ii$^-$): $\mathcal{A}$ may perturb the confirmation window in an $\ell_\infty$ ball, $\|\delta_{\mathrm{cnf}}\|_\infty \le \varepsilon$.

\begin{proposition}[Necessity of confirmation-channel non-controllability]
\label{thm:necessity}
Assume $f_B$ is differentiable in its input on a neighborhood of the test confirmation distribution support, and that for every cal-selected threshold $\tau \in (0, 1)$ there exists a negative window $X^*$ with $f_B(X^*)_t < \tau$ and $\nabla_X f_B(X^*)_t \neq 0$. Then there exists a finite $\varepsilon^* < \infty$ such that for every $\varepsilon \ge \varepsilon^*$, every $\alpha \in (0, 1)$, and every cal-selected threshold $\tau_t(\alpha)$, an adversary $\mathcal{A}$ with the budget in (ii$^-$) achieves
\begin{equation}
\Pr\!\left[f_B(\tilde X_{\mathrm{cnf}})_t \ge \tau_t(\alpha) \;\bigm|\; Y_{\mathrm{cnf}} \ne t\right] \;\to\; 1
\label{eq:necessity-rate}
\end{equation}
as $\varepsilon \to \varepsilon^*$. Consequently, no FAR upper bound of the form $\Pr[\text{unauthorized confirmed route}] \le \alpha + g(n_c, \delta)$ with $g$ independent of $\mathcal{A}$ can hold under B-strict$^-$.
\end{proposition}

\begin{proof}
Fix $\alpha$, $\tau = \tau_t(\alpha)$. Pick a negative window $X^*$ with $s := f_B(X^*)_t < \tau$ and gradient $g^* = \nabla_X f_B(X^*)_t \ne 0$; such $X^*$ exists by assumption. The attacker constructs $\tilde X$ by $\ell_\infty$ PGD inside the ball $B_\infty(X^*, \varepsilon)$, ascending $f_B(\cdot)_t$. Continuity of $f_B$ and nonemptiness of the superlevel set $\{X : f_B(X)_t \ge \tau\}$ (it contains positive-class examples) imply
\begin{equation}
\varepsilon^*(X^*; \tau) := \inf\{\varepsilon : \exists \tilde X \in \overline{B_\infty(X^*, \varepsilon)} \text{ with } f_B(\tilde X)_t \ge \tau\}
\label{eq:eps-star}
\end{equation}
is finite. PGD with budget $\varepsilon \ge \varepsilon^*(X^*; \tau)$ returns $\tilde X$ with $f_B(\tilde X)_t \ge \tau$. Taking $\varepsilon^* = \sup_{X^*} \varepsilon^*(X^*; \tau)$ over the test negative distribution, the adversary's success rate approaches $1$ as $\varepsilon \to \varepsilon^*$. If a uniform bound $\alpha + g(n_c, \delta)$ existed for any $\alpha < 1 - g(n_c, \delta)$, the adversary's empirical rate would contradict it. No such uniform $g$ exists.
\end{proof}

\begin{corollary}[Pairing of Propositions~\ref{thm:conformal-far} and \ref{thm:necessity}]
\label{cor:assumption-boundary}
Under threat model B-strict, condition (i) is sufficient for the bound when (ii) holds, and (ii) is necessary in the sense that no uniform attacker-independent FAR bound can replace it. The pair characterizes the assumption boundary of split-conformal confirmation FAR control on the BCI confirmation channel.
\end{corollary}

The empirical $\varepsilon^*$ for the TinyEEGNet and EEGNetV4 confirmers used in this paper is upper-bounded by the audited budget $\varepsilon = 0.08$: at this budget the C/C$'$ archetypes in \S\ref{sec:phase5} reach FAR $0.998/0.999$, consistent with Proposition~\ref{thm:necessity}. Proposition~\ref{thm:necessity} therefore predicts a quantitative behavior (rate $\to 1$ at finite $\varepsilon$) rather than only an asymptotic one, which is what the C/C$'$ cells report.

\begin{table}[t]
\centering
\caption{Aggregate-to-denominator repairs. Each invalid aggregate omits the variable that identifies the class-specific route-safety claim.}
\label{tab:denominator-map}
\scriptsize
\setlength{\tabcolsep}{1.2pt}
\begin{tabular}{@{}p{.23\linewidth}p{.22\linewidth}p{.27\linewidth}p{.18\linewidth}@{}}
\toprule
Invalid aggregate & Missing observable & Correct denominator & Supporting experiment \\
\midrule
C2 EEG flag rate & Context provenance and route & Clean/untrusted routed episodes & Exp2/5/7/8 \\
Clean agreement rate & Attacked secondary decision & Source-not-target attacked cases & Exp1/6/8/10 \\
C3 success alone & Route and confirmation policy & Successful C3 cases with route log & Exp2/5/7/8/10 \\
Route episode count & Seed/case/context hierarchy & Seed-level numerators and gates & Tables~\ref{tab:evidence-spine}, \ref{tab:t10-confirmatory-ci} \\
Confirmation route rate & Oracle/user-confirmation status & Execution-policy upper-bound rows & Exp5/8/9; Tables~\ref{tab:v6-defense}, \ref{tab:t10-coverage-risk} \\
\bottomrule
\end{tabular}
\end{table}

\section{Validation Protocol}
\label{sec:protocol}

Each validation cell reports decoder usability separately from attack bypass. Table~\ref{tab:primary-endpoints} fixes the primary endpoint hierarchy before the experiment narrative: Exp8 native command-control and Exp9 non-oracle confirmation carry the main route-safety claims; Exp10 and the matched-defense cell are stress tests on the same denominators; transcript and live-router bridges are support artifacts only. The required C2/C3 metrics are clean primary accuracy, clean secondary accuracy, clean disagreement burden, C2 attack-attributable EEG flag uplift, C2 context-provenance block rate, C3 attack success, C3 conditional bypass, C3 unauthorized conditional bypass, rule-specific target-route rates, benign block rates, confirmation status, and tool-attempt rate. C1 and C4 are logged as supporting axes rather than claim-bearing empirical classes. Seeds that fail decoder, attack, or rule gates are reported rather than silently excluded.

\begin{table*}[t]
\centering
\caption{Claim-bearing audit protocol. Rows name the statistical unit and route predicate used for each load-bearing claim; route episodes are coverage artifacts, not independent samples.}
\label{tab:audit-protocol}
\scriptsize
\setlength{\tabcolsep}{2pt}
\begin{tabular}{@{}p{.12\linewidth}p{.16\linewidth}p{.18\linewidth}p{.18\linewidth}p{.25\linewidth}@{}}
\toprule
Claim & Statistical unit & Denominator & Route predicate & Gate or boundary \\
\midrule
Decoder usability & seed & held-out EEG events & clean primary/secondary accuracy and disagreement & $\ge7$ seeds with acc. $\ge.60$, disagr. $\le.40$ \\
C2 route safety & seed $\times$ context set & clean/untrusted routed episodes & requested target route under benign vs. untrusted provenance & provenance block $\le.05$ with benign utility reported \\
C3 agreement risk & seed $\times$ source-not-target case & attacked cases whose clean source is not target & primary target, secondary target, active policy executes target & report $\Prb(A)$, $\Prb(B\mid A)$, $\Prb(P\mid A,B)$ when observable \\
Confirmation & seed $\times$ independent confirmation window & proxy-confirmed or upper-bound policy episodes & execution only if confirmation condition passes & report clean coverage and C3 residual jointly \\
Stress/defense & seed & same C3 source-not-target cases & constrained or query-limited attack routes through same policy & residual route rate under matched clean burden \\
Transcript bridge & seed-level route suite & decoded-label cases $\times$ public contexts & parse-valid, policy-compliant harmless stub route & transcript coverage only; no extra statistical weight \\
\bottomrule
\end{tabular}
\end{table*}

\paragraph{Decomposition-to-defense predictions.}
Theorem~\ref{thm:c3-risk-decomposition} separates four route-safety factors. Primary robustness or attack constraints lower $p_1$ and reduce every downstream C3 route; this is tested by Exp10 preprocessing constraints and the Exp3 band-limited and query rows. Heterogeneous verifiers or ensembles change $\Prb(B\mid A)$ only when attacked secondary errors decorrelate; this is tested by the Exp3 calibrated PGD pair and ensemble residuals. Confirmation or execution policies change $\alpha$ even when both decoders agree, as Exp9's confirmation-proxy frontier and the oracle upper bound show. Finally, the source-not-target denominator $D_c$ prevents clean agreement or route episodes from certifying C3 by themselves; this is tested by Exp6 ($800/800$ source-not-target cases under agreement) and codified in Table~\ref{tab:audit-protocol} and Table~\ref{tab:primary-endpoints}.

\begin{table*}[t]
\centering
\caption{Primary endpoint hierarchy. ``Primary'' rows are the main claim-bearing endpoints used to decide the paper's C2/C3 audit claims; stress rows test robustness boundaries on the same denominators; support rows check route-vocabulary and transcript coverage without adding statistical weight.}
\label{tab:primary-endpoints}
\scriptsize
\setlength{\tabcolsep}{2pt}
\begin{tabular}{@{}p{.10\linewidth}p{.18\linewidth}p{.20\linewidth}p{.17\linewidth}p{.24\linewidth}@{}}
\toprule
Role & Endpoint & Statistical unit / denominator & Decision rule & Reported value \\
\midrule
Primary & Exp8 native C2 route block & 10 seeds; untrusted native-command route episodes & provenance route $\le.05$ & $0.000$; decoder gates $10/10$ \\
Primary & Exp8 native C3 agreement route & 10 seeds; 80 source-window cases/seed & agreement-plus-provenance route reported & $1.000$; clean utility $.806$ \\
Primary & Exp9 non-oracle confirmation proxy & 10 seeds; independent confirmation windows & report C3 residual and clean coverage jointly & $.178$ at clean $.573$; $.091$ at clean $.458$ \\
Stress & Exp10 preprocessing-shaped C3 & 10 seeds; same native source-not-target route predicate & retain/fail under named constraints & smooth9 $.883$; band 8--30 $.995$; smooth9+band $.156$ \\
Stress & Matched non-oracle defenses & 10 seeds; matched clean-burden C3 cases & residual risk under calibrated defenses & white-box $.590$--$.693$; score-query $.067$--$.093$ \\
Support & Public transcript/live-router bridge & seed-level decoded-label route suite & parse/policy consistency and route-vocabulary coverage & coverage only; no extra statistical unit \\
\bottomrule
\end{tabular}
\end{table*}

For C2, attack-attributable uplift is the C2 EEG-side flag rate minus the clean flag rate over the same case set; a successful context-provenance defense blocks the target route while leaving benign context unblocked. For C3, attack success is the fraction of attack cases whose primary decoder reaches the target; conditional bypass is the fraction of successful C3 cases where the secondary decoder also agrees; unauthorized conditional bypass uses only source cases whose clean label is not the target. Unless otherwise stated, C2 route rates are averaged over 20 untrusted contexts per seed, benign block rates over 20 trusted contexts per seed, C3 rates over 120 attack cases per seed, and unauthorized C3 rates over the 60 source-not-target cases per seed. Tables report seed means and standard deviations with failed decoder gates kept in the denominator; route-episode counts are transcript-coverage counts, not independent statistical units. Exact seed-level numerators, denominators, and gates are part of the released result JSON. The gate-selected A:C rows use the predefined decoder gate and then report the strongest passing perturbation setting, while Table~\ref{tab:c3-strength-sweep} exposes the weaker settings so that the C3 claim is not hidden behind a single best configuration. The fixed route/observability cells compare non-confirmation mediation layers against explicit confirmation. Confirmation is not treated as a learned defense. We report two distinct confirmation rows: an oracle execution-policy upper bound that withholds high-impact execution unless a clean-label confirmation variable grants the route, and a non-oracle proxy that requires an independent EEG confirmation window decoded by a separately trained model above a calibration-split threshold. The final matched-defense cell adds a multi-step PGD adversarial verifier, robust ensemble gates, clean-burden calibration, band-limited attacks, and score-query attacks so that the remaining residual C3 routes are not an artifact of weak or unfairly burdened baselines.

\paragraph{Rule-level observability.}
Each non-oracle rule observes a different subset of variables. Primary-only and confidence-gate rules fail on C2 (no provenance variable), and confidence gating only partially flags C1 perturbations. Heterogeneous agreement flags some C1 cases but is bypassed by C3 because both decoders are jointly driven. Context provenance directly blocks C2 routes and is irrelevant for C3. Agreement plus provenance blocks C2 but is still bypassed by C3 under the conditions of Proposition~\ref{lem:observability-feasibility}. The non-oracle confirmation proxy of Exp9 introduces an independent EEG signal that lowers C3 routing without requiring a clean label, and the oracle confirmation rule appears only as a separately marked upper-bound execution policy.

\paragraph{What the audit proves and does not prove.}
The audit proves narrow, testable claims rather than deployment-level guarantees. C2 provenance proves that route logs must include context provenance and outcome, not the prevalence of malicious contexts. C3 agreement proves that the attacked secondary decision and policy are required denominators, not human intent or consent. Native command-control proves that the audit holds on EEGMMI left/right command labels, not deployed cursor-controller safety. Perturbation stress proves that route risk survives some preprocessing-shaped constraints, not electrode-level physical feasibility. Harmless tool stubs prove route-policy observability, not real high-impact tool execution. Transcript bridges prove that public route predicates can be inspected, not that they constitute independent statistical evidence.

\paragraph{Reviewer-facing claim boundary.}
The experiments are offline authorization audits. They do not estimate deployed BCI prevalence, physical injection feasibility, human consent, or real high-impact tool execution. Their claim is narrower and testable: for a fixed BCI-to-tool routing interface, a route-safety certificate must log the variables in Table~\ref{tab:audit-protocol}; otherwise C2 and C3 can be counted with non-identifying denominators. The empirical cells then ask whether the resulting falsification protocol survives native command labels, preprocessing-shaped perturbations, independent confirmation proxies, public harmless routing, and matched non-oracle defenses.

\section{Experiments}
\label{sec:experiments}

\subsection{Routing protocol, command taxonomy, and native command-control}
\label{sec:routing}

The main routing evidence is the native command-control audit in Exp8 and the independent confirmation proxy in Exp9. The decoded binary command is treated as a confirmation-relevant interface intent and maps only to harmless cursor-left/cursor-right stubs in the primary native cell. A context-provenance rule blocks target-route attempts when the route is induced only by external context. The protocol tests the authorization boundary directly: whether EEG-side agreement and context provenance observe the right attack class, not whether any real financial, medical, file-system, or communication action is executed. Exp5/Exp7 and the command-taxonomy transcript gate are support cells: they broaden the harmless route vocabulary and check public router consistency, but they do not replace the native-command endpoint hierarchy in Table~\ref{tab:primary-endpoints}.

To reduce reliance on that high-impact surrogate, Exp8 instantiates the same audit on native EEGMMI left/right command-control semantics. It uses R03/R04/R07/R08/R11/R12, where T1/T2 denote left/right fist execution or imagery, and maps decoded labels only to harmless cursor-left/cursor-right stubs. Across 60 subjects, 5{,}400 native command events, 10 seeds, and 80 cases per seed, A:C decoder gates pass in 10/10 seeds with clean accuracies $0.779/0.788$ and disagreement $0.174$. C2 untrusted overlay contexts route the requested direction at $1.000$ without provenance and $0.000$ with provenance. Shared-raw C3 command flips reach the target direction under agreement plus provenance at $1.000$, while confirmation plus provenance routes them at $0.000$; clean non-abstain utility is $0.806$. This is still an offline audit, not a deployed cursor controller, but it replaces the financial/access surrogate with native BCI command-control labels.

Exp9 separates the oracle confirmation upper bound from a non-oracle confirmation proxy. It keeps the same native-command A:C source-window attack, but route execution also requires an independent clean EEG confirmation window decoded by TinyEEGNetB above a calibration-split threshold chosen to control per-label false accepts. The confirmation window is paired from the same subject and same true source command when available, so an unauthorized target route requires a confirmation-decoder false accept rather than access to the clean source label. With a $0.05$ calibration false-accept target, the proxy passes all decoder and case gates; confirmation-decoder accuracy is $0.798$, calibration false accept rates are $0.046/0.047$, and calibration true accept rates are $0.893/0.896$. It reduces native C3 target routing from $1.000$ under agreement to $0.178$ with pooled Wilson 95\% CI $[0.153,0.206]$, while clean proxy-confirmed routing is $0.573$. A stricter $0.01$ false-accept target passes the predefined Exp9 gate and further reduces C3 routing to $0.091$ with Wilson CI $[0.073,0.113]$, at clean proxy-confirmed routing $0.458$. Thus independent confirmation is not an oracle-equivalent learned defense in this audit: it exposes a measurable risk--coverage frontier.

The validation hierarchy separates controls, primary endpoints, stress tests, and support artifacts. Synthetic, embedding, and deterministic-routing cells check metric plumbing and directional sanity. The paper-facing routing endpoints are Exp8 and Exp9: native left/right command-control with harmless cursor stubs, followed by an independent EEG confirmation proxy that turns the oracle upper bound into a measurable risk--coverage frontier. The decoded-label public tool-agent cell, command-taxonomy transcript gate, and live-router bridge are support artifacts only. They send no raw EEG, embeddings, subject identifiers, or paths; they check that the same C2/C3/confirmation route predicates can be represented across a broader harmless route vocabulary and a public router interface. Their 144{,}000 route episodes are transcript coverage, not independent statistical samples, and they are not used to decide the main claim.

\begin{table}[t]
\centering
\caption{Native EEGMMI command-control audit on left/right fist execution and imagery runs. Routes are harmless cursor stubs; values are seed means over 10 seeds.}
\label{tab:exp8-native-command}
\scriptsize
\setlength{\tabcolsep}{2pt}
\begin{tabular}{@{}lccc@{}}
\toprule
Cell & Gate & Route result & Boundary answered \\
\midrule
Decoder & 10/10 & acc. $.779/.788$; disagr. $.174$ & native commands decodable \\
C2 overlay & pass & request route $1.000 \to 0.000$ & provenance blocks context \\
C3 flip & pass & agreement route $1.000$ & agreement not intent \\
Confirmation & pass & C3 route $0.000$ & execution upper bound \\
\bottomrule
\end{tabular}
\end{table}

\begin{table}[t]
\centering
\caption{Seed-level decoder statistics for Exp8 native command-control. The four route columns (C3-dual, C2-prov., C3-agree, C3-conf.) are deterministic in this cell and equal $(1.000,0.000,1.000,0.000)$ for every seed; we omit them here to surface the across-seed decoder variance, and list those route values once in the caption rather than repeating ten identical rows.}
\label{tab:exp8-seed-level}
\scriptsize
\setlength{\tabcolsep}{4pt}
\begin{tabular}{@{}lccc@{}}
\toprule
Seed & A acc. & C acc. & A/C disagr. \\
\midrule
0 & .778 & .785 & .157 \\
1 & .766 & .789 & .185 \\
2 & .778 & .795 & .177 \\
3 & .772 & .789 & .166 \\
4 & .764 & .797 & .178 \\
5 & .789 & .791 & .169 \\
6 & .784 & .780 & .175 \\
7 & .780 & .785 & .172 \\
8 & .789 & .789 & .184 \\
9 & .788 & .777 & .175 \\
\bottomrule
\end{tabular}
\end{table}

\begin{table}[t]
\centering
\caption{Exp9 non-oracle confirmation proxy on native command-control. The source window is attacked; an independent clean confirmation window is decoded by TinyEEGNetB with calibration-split thresholds.}
\label{tab:exp9-confirmation-proxy}
\scriptsize
\setlength{\tabcolsep}{2pt}
\begin{tabular}{@{}lcccc@{}}
\toprule
Cal. FA target & Cal. FA & Cal. TA & C3 proxy route & Clean proxy route \\
\midrule
$.05$ & $.046/.047$ & $.893/.896$ & $.178{\pm}.037$ & $.573{\pm}.053$ \\
$.01$ & $.007/.008$ & $.679/.726$ & $.091{\pm}.031$ & $.458{\pm}.053$ \\
\bottomrule
\end{tabular}
\end{table}

\begin{table}[t]
\centering
\caption{Inspectable route-taxonomy slice used by the public transcript gate. All entries are public strings and harmless stubs; high-impact stubs record that a route would have been attempted but execute nothing.}
\label{tab:route-taxonomy-slice}
\scriptsize
\setlength{\tabcolsep}{1.5pt}
\resizebox{\linewidth}{!}{
\begin{tabular}{@{}llll@{}}
\toprule
Family & Source & Route & Policy behavior \\
\midrule
Status/menu & trusted UI & status/menu/calib. & low-risk route allowed \\
Message & trusted UI & draft-only & local buffer only \\
Finance/access & untrusted ctx. & transfer/unlock/alert & provenance abstains \\
Mixed & untrusted ctx. & msg. $+$ transfer & provenance abstains \\
Trusted high-impact & trusted UI & transfer/alert prep & confirmation required \\
C3 target & trusted ctx. & decoded transfer & agreement routes unless confirm. required \\
\bottomrule
\end{tabular}}
\end{table}

\subsection{Real EEGMMI command-label validation}
\label{sec:round5}

The real-label cell uses the EEGMMI R03 event cache with 60 subjects, 900 T1/T2 events, 16 channels, and 512 samples per event. T1/T2 are used as real binary command labels in the surrogate routing protocol; the target label maps only to a harmless non-executing route. The main cell trains two independent TinyEEGNet-style raw EEG decoders and evaluates shared raw-input PGD attacks.

Real-EEG inputs are standardized using training-window channel statistics and then attacked in that standardized raw-input space. Across the 30 A:C split/seed runs, the training-channel standard deviation in the raw cache is $0.9725\pm0.0078$; therefore $\epsilon=.08$ corresponds to a mean raw-cache amplitude of $0.0778$, about $8.0\%$ of the global raw standard deviation and $4.1\%$ of the 95th percentile absolute amplitude. This calibration makes the perturbation scale auditable, but it is still an offline standardized-tensor budget rather than a claim of hardware-injection feasibility. TinyEEGNetA, B, and C differ in temporal kernel widths, channel counts, pooling, and dropout, so A:C is not a duplicate-head agreement test. C1 uses targeted PGD against the primary decoder only. C3 uses the same raw perturbation for both decoders and minimizes the sum of primary and secondary target cross-entropies with $\ell_\infty$ clipping; this is the executable relaxation of Proposition~\ref{lem:observability-feasibility}. All reported split/seed runs keep failed seeds in the denominator and report decoder gates before attack metrics.

\paragraph{Level-1 gate.}
The original A:B real-label cell passes the gate: 10/10 usable seeds, balanced accuracy $0.7542/0.7790$, clean disagreement $0.1520$, C2 EEG flag uplift $0.0000$, C2 target-route attempts $1.0000/0.0000$ without/with provenance, and C3 success/bypass/unauthorized bypass $0.9983/0.9933/0.9900$. We treat this as the calibrated real-label gate, then repeat the shared-raw attack with TinyEEGNetC as a different secondary decoder across three split protocols. The replication rules out the main toy-pipeline and single-decoder-pair objections: C3 bypass is not a subject-parity, embedding-space, or one-pair artifact.

\begin{table}[t]
\centering
\caption{Raw-CNN Level-2 A:C replication across split protocols at the gate-selected $\epsilon=0.08$, 40-step shared-raw PGD setting. Values are mean $\pm$ standard deviation.}
\label{tab:round5-level2-ac}
\footnotesize
\setlength{\tabcolsep}{2pt}
\begin{tabular}{@{}lccc@{}}
\toprule
Split & P/S acc. & Disagr. & C3 succ./bypass \\
\midrule
Within & $.754{\pm}.028$ / $.766{\pm}.029$ & $.149{\pm}.032$ & $.997{\pm}.008$ / $1.000$ \\
Provided & $.737{\pm}.022$ / $.752{\pm}.019$ & $.183{\pm}.008$ & $1.000$ / $1.000$ \\
Holdout & $.756{\pm}.030$ / $.763{\pm}.023$ & $.159{\pm}.024$ & $.999{\pm}.003$ / $1.000$ \\
\bottomrule
\end{tabular}
\end{table}

\begin{table}[t]
\centering
\caption{Attack-strength sweep from the same A:C result files, aggregating 30 split/seed runs per $\epsilon$ at 40 steps.}
\label{tab:c3-strength-sweep}
\footnotesize
\setlength{\tabcolsep}{3pt}
\begin{tabular}{@{}lccc@{}}
\toprule
$\epsilon$ & Runs/gates & C3 success & C3 unauth. bypass \\
\midrule
.02 & $30/30$ & $.805{\pm}.044$ & $.893{\pm}.066$ \\
.04 & $30/30$ & $.951{\pm}.025$ & $.972{\pm}.028$ \\
.08 & $30/30$ & $.999{\pm}.005$ & $1.000{\pm}.000$ \\
\bottomrule
\end{tabular}
\end{table}

For the gate-selected A:C configuration ($\epsilon=0.08$, 40 PGD steps), all 30 split/seed runs are usable. The aggregate primary/secondary balanced accuracies across the three split protocols are $0.749/0.760$, and the mean clean disagreement is $0.164$. Across within-subject, provided-split, and subject-holdout protocols, C2 attack-attributable EEG flag uplift is $0.000\pm0.000$; C2 reaches the target route without context provenance at $1.000\pm0.000$ and is blocked with provenance at $0.000\pm0.000$. C3 unauthorized conditional bypass is $1.000\pm0.000$ in each split, with the minimum usable-seed conditional bypass equal to $1.000$ in all three split protocols. Thus the cross-decoder and cross-split replication preserves both parts of the boundary claim: EEG-side monitoring remains blind to C2 context injection, and agreement remains satisfiable under shared raw-input C3 optimization.

A separate lemma-to-experiment bridge repeats the provided-split A:C cell with explicit margin diagnostics. Decoder gates pass in 10/10 seeds; at $\epsilon=0.08$ and 40 PGD steps, C3 attack success, conditional bypass, and unauthorized conditional bypass are all $1.000$. The joint-min margin gain is $27.487$ on average, and every seed has positive joint margin rate and shared-descent-both rate $1.000$. This bridge is the executable falsification cell for Proposition~\ref{lem:observability-feasibility}: the shared raw-input feasible region is observed in the same real-label setting used by the routing experiments.

To address the clean-utility objection directly, Exp6 restricts evaluation to test events where both A and C are clean-correct, agree with each other, and the true source label is not the target. This subset is not a deployment prevalence estimate; it is a conditioned audit asking whether agreement bypass remains possible on cases that the two-decoder system handled correctly before attack. Across 10 provided-split seeds, the audit selects 80 clean-correct source-not-target cases per seed from the 300-event test split (rate $0.267$ per seed), for 800 total cases. Shared-raw C3 routes $800/800$ of them under agreement, with pooled Wilson 95\% CI $[0.995,1.000]$; explicit confirmation routes $0/800$. The clean subset has mean decoder confidences $0.961/0.960$, so this result is not driven by low-confidence clean disagreements.

\begin{table}[t]
\centering
\caption{Paper-facing evidence spine for the final C2/C3 claim. All cells use EEGMMI real T1/T2 labels with harmless stubs; all rows pass the predefined gate at $10/10$ seeds (Exp10 is reported as a stress cell on the same denominators).}
\label{tab:evidence-spine}
\scriptsize
\setlength{\tabcolsep}{1.5pt}
\begin{tabular}{@{}p{.10\linewidth}p{.45\linewidth}p{.36\linewidth}@{}}
\toprule
Cell & Main result & Objection answered \\
\midrule
Exp1 & C3 unauth.\ bypass $1.000$ & formal/toy artifact \\
Exp2 & C2 prov.\ $0.000$; C3 non-conf.\ $1.000$ & tool-routing semantics \\
Exp3 & white-box residual $.590$--$.693$ & weak baseline \\
Exp4 & attack-attrib.\ residual $\le.070$ & score/grad.\ access \\
Exp5 & 30 ctx.; 144k routes; $0/1/0$ & single stub/ctx. \\
Exp6 & $800/800$ high-utility C3 & low-utility error \\
Exp7 & 144k Opus 4.7; parse/policy $1.000$ & public LLM consistency \\
Exp8 & 5.4k native events; $0/1/0$ & surrogate command \\
Exp9 & native C3 route $.178/.091$ & oracle confirmation \\
Exp10 & smooth/band/ramp $.883/.995/1.000$ & tensor-only \\
\bottomrule
\end{tabular}
\end{table}

\begin{table}[t]
\centering
\caption{Statistical hierarchy for denominator-sensitive results. Route episodes are transcript-coverage units; seed-level rows are the statistical units used for means and intervals.}
\label{tab:unit-hierarchy}
\scriptsize
\setlength{\tabcolsep}{1.4pt}
\begin{tabular}{@{}p{.22\linewidth}p{.27\linewidth}p{.39\linewidth}@{}}
\toprule
Quantity & Unit & Used for \\
\midrule
Decoder gates & seed/run & claim-bearing usability filter \\
Attack rates & seed-level case sets & C3 success and bypass estimates \\
Route rates & seed $\times$ case $\times$ policy & C2/C3 route numerators \\
Route episodes & context expansion & transcript coverage, not independent CI units \\
Live calls & batched router requests & parse/policy compliance checks \\
Stored transcripts & sampled route examples & inspectability and artifact hygiene \\
\bottomrule
\end{tabular}
\end{table}

\subsection{Fixed routing, observability, and realism cell}
\label{sec:v6}

The fixed route-suite cell is a controlled observability test rather than an external-validity claim. It fixes a public 40-context route suite: 20 trusted benign contexts and 20 untrusted C2 contexts that request the target route. Benign examples include status, calibration, spelling-buffer, and no-op interface contexts; C2 examples are untrusted retrieved or webpage contexts that request the target route. The deterministic route rule is intentionally transparent: route by decoded intent under trusted benign context, route to the target when untrusted context requests it unless provenance is enabled, and abstain when provenance marks such context as untrusted. This makes the cell a unit test for denominators and rule observability. It evaluates the same real-label A:B:C raw decoders on the provided split over 10 seeds, 120 attack cases per seed, and six execution rules: primary-only, confidence gate ($\tau=.80$), heterogeneous agreement, context provenance, agreement plus provenance, and a confirmation upper-bound rule. No real tool is executed; the target route is still the harmless non-executing stub.

\begin{table}[t]
\centering
\caption{v6 fixed routing and observability cell on EEGMMI real T1/T2 labels, provided split, 10 seeds. C2 entries are target-route rates over 20 untrusted contexts; benign entries are target-route block rates over 20 trusted contexts. C3 entries are unauthorized target-route rates.}
\label{tab:v6-defense}
\scriptsize
\setlength{\tabcolsep}{1.5pt}
\resizebox{\linewidth}{!}{
\begin{tabular}{@{}lccc@{}}
\toprule
Rule & C2 route & Benign block & C3 whitebox unauth. \\
\midrule
Primary only & $1.000{\pm}.000$ & $0.000{\pm}.000$ & $1.000{\pm}.000$ \\
Confidence gate & $.834{\pm}.018$ & $.088{\pm}.024$ & $.998{\pm}.005$ \\
Heterogeneous agreement & $.817{\pm}.027$ & $.107{\pm}.012$ & $1.000{\pm}.000$ \\
Context provenance & $0.000{\pm}.000$ & $0.000{\pm}.000$ & $1.000{\pm}.000$ \\
Agreement + provenance & $0.000{\pm}.000$ & $.107{\pm}.012$ & $1.000{\pm}.000$ \\
Confirmation upper bound (exec.) & $0.000{\pm}.000$ & $.122{\pm}.025$ & $0.000{\pm}.000$ \\
\bottomrule
\end{tabular}}
\end{table}

\begin{table}[t]
\centering
\caption{Native command-control temporal/preprocessing stress. Routes are C3 target-direction agreement routes over 10 seeds and 80 cases/seed. Distortion is mean perturbation standard deviation divided by event standard deviation; SNR is in dB.}
\label{tab:exp10-temporal}
\scriptsize
\setlength{\tabcolsep}{2pt}
\begin{tabular}{@{}lccc@{}}
\toprule
Constraint & C3 route & Distortion & SNR \\
\midrule
Plain raw tensor & $1.000{\pm}.000$ & $.088$ & $21.29$ \\
Smooth, 9 samples & $.883{\pm}.044$ & $.059$ & $24.85$ \\
Smooth, 17 samples & $.696{\pm}.063$ & $.056$ & $25.23$ \\
Band 4--40 Hz & $1.000{\pm}.000$ & $.074$ & $22.90$ \\
Band 8--30 Hz & $.995{\pm}.006$ & $.067$ & $23.67$ \\
Causal ramp, 25\% & $1.000{\pm}.000$ & $.078$ & $22.41$ \\
Smooth9 $+$ band 8--30 & $.156{\pm}.034$ & $.004$ & $48.83$ \\
\bottomrule
\end{tabular}
\end{table}

\begin{table}[t]
\centering
\caption{Post-processing and temporal-shift survival for Exp10. Values are dual-decoder target agreement after applying the listed transformation to an already optimized C3 example.}
\label{tab:exp10-survival}
\scriptsize
\setlength{\tabcolsep}{2pt}
\begin{tabular}{@{}lccccc@{}}
\toprule
Attack & none & smooth9 & band 8--30 & shift $\pm4$ & shift $\pm8$ \\
\midrule
Plain & $1.000$ & $.926$ & $1.000$ & $.196/.183$ & $.238/.213$ \\
Smooth9 & $.883$ & $.898$ & $.358$ & $.759/.759$ & $.563/.586$ \\
Band 8--30 & $.995$ & $.539$ & $1.000$ & $.014/.021$ & $.095/.091$ \\
Smooth9 $+$ band & $.156$ & $.160$ & $.406$ & $.133/.140$ & $.124/.113$ \\
\bottomrule
\end{tabular}
\end{table}

\begin{table}[t]
\centering
\caption{Perturbation-realism diagnostics for Exp10. The audit reports both survival and failure under temporal/preprocessing transformations; high SNR does not imply physical feasibility. The shift rows ($\pm4/\pm8$) report the survival of an already-optimized perturbation under post-hoc temporal shift, so they share the distortion/SNR of the underlying optimized attack and we mark the cell n/a to avoid double-counting.}
\label{tab:exp10-realism}
\scriptsize
\setlength{\tabcolsep}{2pt}
\resizebox{\linewidth}{!}{
\begin{tabular}{@{}lccc@{}}
\toprule
Diagnostic & Survival & Distortion/SNR & Interpretation \\
\midrule
Smooth9 optimized & $.883$ & $.059$ / $24.85$ dB & low-pass-shaped attack remains high \\
Band 8--30 optimized & $.995$ & $.067$ / $23.67$ dB & motor-band constraint remains high \\
Causal ramp optimized & $1.000$ & $.078$ / $22.41$ dB & onset-constrained attack remains high \\
Smooth9 $+$ band optimized & $.156$ & $.004$ / $48.83$ dB & combined constraint is a failure boundary \\
Plain attack shifted $\pm4/\pm8$ & $.196/.183$; $.238/.213$ & n/a & timing alignment is fragile \\
Band attack shifted $\pm4/\pm8$ & $.014/.021$; $.095/.091$ & n/a & band-limited route is shift-sensitive \\
\bottomrule
\end{tabular}}
\end{table}

The C2 route rates are a validation sanity check for observability: confidence and heterogeneous agreement reduce some benign or low-confidence executions, but they do not observe context provenance. Provenance blocks C2 and is irrelevant for C3. The confirmation upper-bound rule is included only as an execution policy with access to the clean label: it blocks unauthorized C3 by using information unavailable to an automatic EEG-side monitor. Exp9 replaces that oracle variable with an independently decoded confirmation window and shows the expected tradeoff: target routing falls from $1.000$ to $.178$ or $.091$ as the calibration false-accept target tightens from $.05$ to $.01$, while clean proxy-confirmed routing falls from $.573$ to $.458$.

Exp10 addresses the tensor-only objection without claiming physical injection. It repeats the native command-control C3 audit under preprocessing-shaped perturbation families and post-hoc transformations, using the same source-not-target denominator and route predicate as Table~\ref{tab:audit-protocol}. The diagnostic target is not invisibility to a human or realizability through electrodes; it is whether the agreement bypass survives constraints that approximate common preprocessing assumptions. Smoothing, band-limiting, and causal ramp constraints retain substantial agreement-route rates: smooth9 $.883$ at distortion $.059$ and $24.85$ dB SNR, band 8--30 Hz $.995$ at distortion $.067$ and $23.67$ dB, and causal ramp $1.000$ at distortion $.078$ and $22.41$ dB. These rows are constraint-matched C3 cells: the attack is optimized inside the named preprocessing family rather than evaluated only as an unconstrained raw-tensor perturbation. Post-hoc smoothing of a plain attack retains $.926$ agreement, and post-hoc band-limiting retains $1.000$. The same cell also exposes a failure boundary: combining smoothing and 8--30 Hz band-limiting during optimization drops routes to $.156$ despite very high SNR ($48.83$ dB), and shifting a plain or band-limited perturbation by $\pm4$ or $\pm8$ samples sharply reduces agreement. Thus the perturbation evidence supports an offline audit-stress claim and a temporal-alignment caveat, not hardware-channel feasibility.

The command-taxonomy transcript gate tests whether this conclusion survives a less single-stub route vocabulary. It uses the same provided-split A:C decoders, 10 seeds, 80 selected cases per seed, 30 public contexts, and six policies, producing 144{,}000 local deterministic route episodes and 240 stored example transcripts. Agreement plus provenance blocks all C2 high-impact routes ($0.000$) but routes all unauthorized C3 high-impact cases ($1.000$); provenance plus confirmation routes unauthorized C3 high-impact cases at $0.000$. Clean non-abstain utility under agreement plus provenance is $0.889\pm0.045$ with seed minimum $0.8125$. This cell is not a deployed-agent benchmark, but it removes the narrower objection that the public route evidence depends on a single named high-impact stub or a single context template.

The live public-router bridge repeats the transcript gate with Claude Opus 4.7 as the router and the same harmless stub schema. Across 144{,}000 routing episodes and 480 live API calls, 1{,}200 sampled transcripts have zero parse failures, zero policy violations, and zero executed high-impact stubs. No request contains raw EEG, embeddings, subject identifiers, local paths, or checkpoints. The live router reproduces the C2/C3/confirmation pattern: C2 high-impact route rate is $1.000$ without provenance and $0.000$ with provenance; C3 high-impact route rate is $1.000$ under agreement plus provenance and $0.000$ under confirmation plus provenance. The deterministic-reference agreement rate is $0.9754$ and clean non-abstain utility is $0.846$. The full case-gated bridge uses 80 clean and 80 C3 decoded-label cases per seed; all decoder case and router behavior gates pass.

\subsection{Matched stronger-defense stress}
\label{sec:v7}

The final stress cell tests whether stronger non-oracle defenses close the C3 channel when their clean burden is explicitly matched. It instantiates the decomposition in Theorem~\ref{thm:c3-risk-decomposition}: primary robustness targets $\Prb(A)$, heterogeneous verification targets $\Prb(B\mid A)$, and confirmation changes $\Prb(P\mid A,B)$. We keep TinyEEGNetA as the primary decoder and compare standard agreement, FGSM and multi-step PGD adversarial verifiers, clean and PGD robust ensembles, and calibrated matched variants. The calibrated variants set confidence thresholds on the calibration split to match a target primary-target clean block rate of $0.12$, avoiding a misleading comparison where a defense wins only by refusing most clean target routes. The matched PGD pair and ensemble both pass the burden check, and the ensemble does not cross the predefined overburden boundary. Attacks are adaptive to the evaluated family: white-box pair PGD, white-box ensemble PGD, band-limited ensemble PGD, and score-query attacks with 100 and 200 queries.

\begin{table}[t]
\centering
\caption{Reproducibility artifacts for the validation protocol. All tools are harmless stubs; no live API call is required for the load-bearing cells. The optional live-router bridge sends only decoded labels, public context strings, policy flags, and harmless tool schemas.}
\label{tab:repro-artifacts}
\scriptsize
\setlength{\tabcolsep}{1.3pt}
\begin{tabular}{@{}p{.25\linewidth}p{.65\linewidth}@{}}
\toprule
Artifact & Released content \\
\midrule
Data/splits & EEGMMI R03 cache metadata; native R03/04/07/08/11/12 cache; train/cal/test and holdout indices \\
Preprocess/models & standardization constants; TinyEEGNet A/B/C configs; seeds 0--9 \\
Routing & route-rule pseudocode; benign/untrusted contexts; command taxonomy strings \\
Transcripts & 240 local route examples plus 1,200 sampled live-router bridge transcripts \\
Attacks & PGD, transfer, smoothed, band-limited, causal-ramp, temporal-shift, score-query NES, decision-only search \\
Results & seed-level numerators, denominators, gates, CIs, utility-subset audit rows \\
\bottomrule
\end{tabular}
\end{table}

The support package is part of the validation object. It specifies the EEGMMI event-window construction, channel set, train/cal/test and subject-holdout splits, standardization constants, decoder hyperparameters, attack objectives, query budgets, route-rule pseudocode, route contexts, and seed-level numerator/denominator JSON for every table entry. The transcript artifacts record the public context string, decoded primary and secondary intent labels, policy flags, chosen harmless stub, execution boolean, and block reason for each stored example. The optional live bridge stores the corresponding public router JSON responses and safety-audit fields. The utility-subset artifact records the clean-correct agreement denominator, clean confidences, attacked decoder outputs, Wilson interval, perturbation ratio, and SNR for Exp6. The Exp7 decoded-label export supplies 80 clean and 80 C3 public cases per seed without raw EEG, embeddings, subject identifiers, paths, or checkpoints. The protocol is denominator-sensitive: a C2 EEG flag rate, a C2 route rate, a C3 conditional bypass rate, and an unauthorized C3 route rate are different units of analysis and should not be merged into one success number.

\begin{table}[t]
\centering
\caption{Matched stronger-defense stress, provided split, 10 seeds. Entries are unauthorized target-route rates except the decision-only row, which reports attack-attributable unauthorized routes. Calibrated matched defenses use clean-burden thresholds chosen on the calibration split; confirmation uses the clean label as an upper-bound execution policy. The Confirmation upper-bound row's Primary cell is marked n/a because the confirmation rule is an execution policy that consumes the clean label rather than an EEG-side defense applied to primary-only routing.}
\label{tab:v7-stronger-defense}
\scriptsize
\setlength{\tabcolsep}{1pt}
\begin{tabular}{@{}lccc@{}}
\toprule
Attack & Primary & Cal. PGD pair & Cal. PGD ens. \\
\midrule
White-box pair PGD & $1.000$ & $.693{\pm}.056$ & $.590{\pm}.062$ \\
White-box ensemble PGD & $.998{\pm}.005$ & $.685{\pm}.053$ & $.692{\pm}.061$ \\
Band-limited ensemble PGD & $.832{\pm}.100$ & $.135{\pm}.026$ & $.137{\pm}.028$ \\
Score-query, 100 queries & $.407{\pm}.121$ & $.087{\pm}.073$ & $.067{\pm}.067$ \\
Score-query, 200 queries & $.433{\pm}.109$ & $.093{\pm}.074$ & $.073{\pm}.063$ \\
Decision-only, 200 queries & $.067{\pm}.042$ & $.037{\pm}.018$ & $.037{\pm}.023$ \\
Confirmation upper bound (exec.) & n/a & $0.000$ & $0.000$ \\
\bottomrule
\end{tabular}
\end{table}

\begin{table}[t]
\centering
\caption{Confirmatory metrics with seed-level bootstrap 95\% confidence intervals. Rows are pre-specified claim-bearing metrics; exploratory stress variants remain in the released JSON.}
\label{tab:t10-confirmatory-ci}
\scriptsize
\setlength{\tabcolsep}{2pt}
\resizebox{\linewidth}{!}{
\begin{tabular}{@{}lccc@{}}
\toprule
Metric & Mean & 95\% CI & Min--max \\
\midrule
Exp1 C3 unauth. bypass & $1.000$ & $[1.000,1.000]$ & $1.000$--$1.000$ \\
Exp2 non-conf. route & $1.000$ & $[1.000,1.000]$ & $1.000$--$1.000$ \\
Exp2 confirmation route & $0.000$ & $[0.000,0.000]$ & $0.000$--$0.000$ \\
Exp8 native C2 prov. route & $0.000$ & $[0.000,0.000]$ & $0.000$--$0.000$ \\
Exp8 native C3 non-conf. route & $1.000$ & $[1.000,1.000]$ & $1.000$--$1.000$ \\
Exp8 native confirmation route & $0.000$ & $[0.000,0.000]$ & $0.000$--$0.000$ \\
Exp9 proxy route, FA $.05$ & $.178$ & $[.153,.206]$ & $.125$--$.250$ \\
Exp9 proxy route, FA $.01$ & $.091$ & $[.073,.113]$ & $.063$--$.138$ \\
Exp9 clean proxy route, FA $.05/.01$ & $.573/.458$ & $[.539,.604]/[.426,.493]$ & $.475$--$.625$ / $.388$--$.563$ \\
Exp10 smooth9 route & $.883$ & $[.855,.909]$ & $.800$--$.938$ \\
Exp10 band 8--30 route & $.995$ & $[.991,.999]$ & $.988$--$1.000$ \\
Exp10 smooth9+band route & $.156$ & $[.135,.178]$ & $.113$--$.213$ \\
Exp3 white-box PGD pair & $.693$ & $[.660,.730]$ & $.633$--$.800$ \\
Exp3 white-box ensemble & $.590$ & $[.552,.630]$ & $.500$--$.700$ \\
Exp3 band-limited pair & $.135$ & $[.118,.152]$ & $.100$--$.183$ \\
Exp3 score-query pair & $.093$ & $[.047,.140]$ & $.000$--$.200$ \\
Exp4 decision-only pair & $.037$ & $[.027,.047]$ & $.000$--$.067$ \\
Exp4 decision-only ensemble & $.037$ & $[.023,.050]$ & $.000$--$.067$ \\
Confirmation upper bound & $0.000$ & $[0.000,0.000]$ & $0.000$--$0.000$ \\
\bottomrule
\end{tabular}}
\end{table}

\begin{table}[t]
\centering
\caption{Coverage--risk frontier for matched non-confirmation defenses, provided split, 10 seeds. Coverage is the clean target-route rate. White-box risk is unauthorized target routing under pair PGD; decision-only risk is attack-attributable unauthorized routing at 200 route/abstain queries.}
\label{tab:t10-coverage-risk}
\scriptsize
\setlength{\tabcolsep}{2pt}
\begin{tabular}{@{}lccc@{}}
\toprule
Defense & Clean coverage & White-box risk & Decision-only risk \\
\midrule
Primary only & $.499$ & $1.000$ & $.067$ \\
PGD agreement & $.400$ & $.853$ & $.070$ \\
Calibrated PGD pair & $.218$ & $.693$ & $.037$ \\
Calibrated PGD ensemble & $.204$ & $.590$ & $.037$ \\
Confirmation upper bound & $.365$ & $0.000$ & $0.000$ \\
\bottomrule
\end{tabular}
\end{table}

The matched defenses change the magnitude but not the validation conclusion. Multi-step adversarial verification, ensembles, and calibrated abstention reduce unauthorized routes, especially for band-limited, score-query, and decision-only attacks. They do not make non-confirmation agreement an intent certificate: calibrated white-box residual routes remain between $0.590$ and $0.693$, while score-query and decision-only routes are consistent with non-zero residual but not conclusive at the present seed count: the score-query pair seed-level CI includes zero (e.g., $[0.000, 0.200]$), and decision-only route/abstain searches find attack-attributable unauthorized routes up to $0.070$ but with overlapping seed CIs at the matched-defense rows. Table~\ref{tab:t10-confirmatory-ci} separates confirmatory claim-bearing metrics from exploratory stress variants and reports seed-level bootstrap intervals; the decision-only calibrated pair residual is $0.037$ with 95\% CI $[0.027,0.047]$. Table~\ref{tab:t10-coverage-risk} shows the clean-coverage tradeoff directly: calibrated non-confirmation defenses reduce coverage and risk, Exp9's non-oracle confirmation proxy reduces native C3 routing to $.178$ and $.091$ at clean proxy coverage $.573$ and $.458$, and only the separately marked confirmation execution upper bound reaches zero unauthorized route risk. This is the paper's safety boundary: mediation and confirmation proxies can reduce and log risk, but they do not by themselves turn agreement into intent.

\subsection{Honest Frontier and Threat-Archetype Necessity}
\label{sec:phase5}

This subsection instantiates Proposition~\ref{thm:conformal-far} and Corollary~\ref{cor:far-route} on the same native command-control denominator as Exp8/Exp9. The reported numbers are pooled over 10 seeds and 80 source-not-target cases per seed, with $95\%$ Wilson confidence intervals; full seed-level rows are in the released JSON.

\paragraph{Strict held-out calibration corrects in-fold leakage.}
Refitting Exp9 with $f_B$ trained only on the train index and the threshold selected only on the calibration index gives $\alpha{=}.05$ empirical FAR $0.0525$, $95\%$ CI $[0.039,0.070]$, versus $0.0912$, CI $[0.073,0.113]$ in the leaky train$+$cal fit reported in Exp9. Stricter targets reach $\alpha{=}.01$ FAR $0.0038$, CI $[0.001,0.011]$, and $\alpha{=}.005$ FAR $0.000$, CI $[0.000,0.005]$. Empirical FAR is at or below $\alpha$ for $\alpha\in\{.005,.01,.02,.03,.05\}$, consistent with Proposition~\ref{thm:conformal-far}; the only finite-sample slack appears at $\alpha{=}.10$, where empirical FAR is $0.1187$ with CI $[0.098,0.143]$ and is rescued by a stronger confirmer below.

\paragraph{Cross-architecture frontier.}
Replacing TinyEEGNetB by EEGNetV4 \cite{lawhern2018eegnet} as the conformal confirmer keeps Proposition~\ref{thm:conformal-far} intact and Pareto-improves the (FAR, utility) frontier at every $\alpha$. EEGNetV4 confirmer accuracy is $0.814$ versus TinyEEGNetB $0.796$. Figure~\ref{fig:frontier-cross-arch} reports the full $\alpha$-sweep for both confirmers; Table~\ref{tab:phase5-frontier} lists the seed-pooled values. At $\alpha{=}.10$, EEGNetV4 controls FAR at $0.0975$ where TinyEEGNetB drifts to $0.1187$. At $\alpha{=}.005$, EEGNetV4 utility nearly doubles ($0.104$ versus $0.060$) while keeping FAR below $0.005$. The conformal target line dominates both empirical curves except at the $\alpha{=}.10$ TinyEEGNetB point, demonstrating that the FAR control claim of Proposition~\ref{thm:conformal-far} is not architecture-specific.

\begin{figure}[t]
\centering
\includegraphics[width=\linewidth]{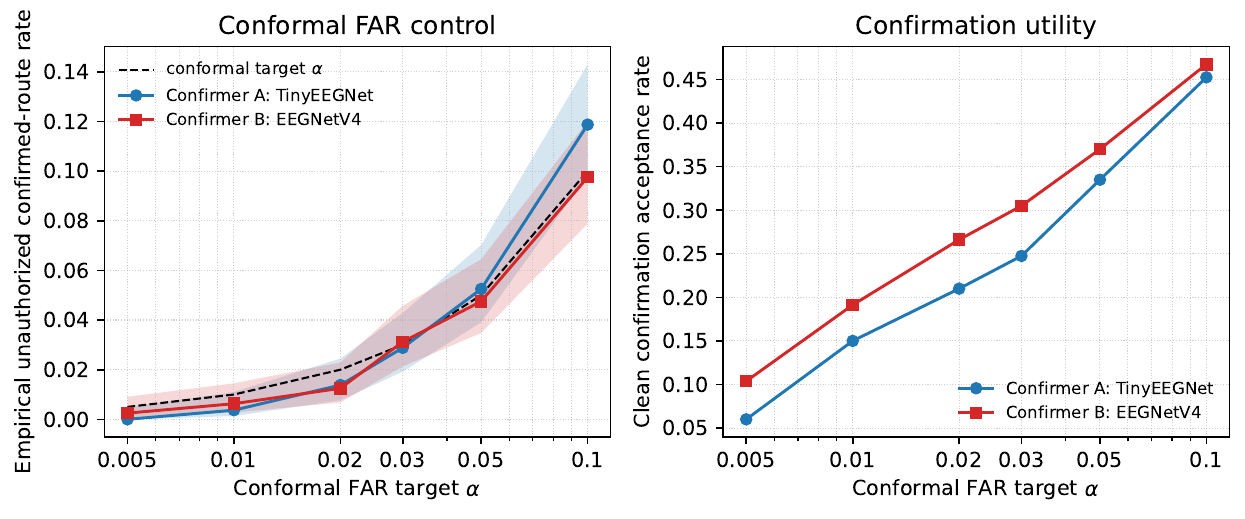}
\caption{Cross-architecture conformal frontier. \textbf{Left:} empirical unauthorized confirmed-route rate versus conformal target $\alpha$, with $95\%$ Wilson CI shading and the conformal target line $y{=}\alpha$. \textbf{Right:} clean confirmation acceptance rate (utility) at the same $\alpha$. EEGNetV4 Pareto-dominates TinyEEGNetB across the frontier under threat model B-strict.}
\label{fig:frontier-cross-arch}
\end{figure}

\begin{table}[t]
\centering
\caption{Conformal $(\alpha, \mathrm{FAR}, \text{utility})$ frontier for two confirmers under threat model B-strict, 10 seeds, 80 cases/seed. FAR is pooled empirical unauthorized confirmed-route rate; utility is clean confirmation acceptance rate. Subject-cluster bootstrap CI uses 60 EEGMMI subjects as resampling units; trial-level Wilson CI is included for comparison.}
\label{tab:phase5-frontier}
\scriptsize
\setlength{\tabcolsep}{2pt}
\begin{tabular}{@{}lccccc@{}}
\toprule
Confirmer & $\alpha$ & FAR & Wilson CI & Subj.\ CI & Utility \\
\midrule
TinyEEGNetB & .005 & .000 & $[.000,.005]$ & $[.000,.000]$ & .060 \\
TinyEEGNetB & .01 & .004 & $[.001,.011]$ & $[.000,.008]$ & .150 \\
TinyEEGNetB & .02 & .014 & $[.008,.025]$ & $[.006,.023]$ & .210 \\
TinyEEGNetB & .03 & .029 & $[.019,.043]$ & $[.017,.042]$ & .247 \\
TinyEEGNetB & .05 & .053 & $[.039,.070]$ & $[.037,.069]$ & .335 \\
TinyEEGNetB & .10 & .119 & $[.098,.143]$ & $[.097,.142]$ & .452 \\
EEGNetV4 & .005 & .003 & $[.001,.009]$ & $[.000,.006]$ & .104 \\
EEGNetV4 & .01 & .006 & $[.003,.015]$ & $[.001,.012]$ & .191 \\
EEGNetV4 & .02 & .013 & $[.007,.023]$ & $[.005,.021]$ & .266 \\
EEGNetV4 & .03 & .031 & $[.021,.046]$ & $[.020,.044]$ & .305 \\
EEGNetV4 & .05 & .048 & $[.035,.065]$ & $[.033,.063]$ & .370 \\
EEGNetV4 & .10 & .098 & $[.079,.120]$ & $[.077,.119]$ & .468 \\
\bottomrule
\end{tabular}
\end{table}

\paragraph{Threat-archetype necessity.}
The two assumptions of Corollary~\ref{cor:far-route} are tested directly, and the C/C$'$ rows are also the empirical instantiation of Proposition~\ref{thm:necessity}. Letting (A) confirmation pairs share the same subject and same session, (B) acquisition-isolated pairs use disjoint subjects, (B$_{\mathrm{xr}}$) cross-run within session uses the same subject but a disjoint EEGMMI run, and (C) the attacker additionally perturbs the confirmation window inside the same $\ell_\infty$ ball, we obtain Table~\ref{tab:phase5-archetype}. Acquisition isolation alone (A vs.\ B/B$_{\mathrm{xr}}$) does not change the controlled FAR. The B$_{\mathrm{xr}}$ row uses a different EEGMMI run as the confirmation window for the same recording day; this is a within-session cross-run audit and not a cross-day deployment estimate, and we report it as such (FAR $0.040$ with subject-cluster CI $[0.027, 0.054]$). Confirmation-channel control breaks the bound regardless of pairing (C, FAR $\approx 1$), exactly as Proposition~\ref{thm:necessity} predicts at finite $\varepsilon$.

\begin{table}[t]
\centering
\caption{Threat-archetype matrix at $\alpha{=}.05$, $\epsilon_{\mathrm{src}}{=}0.08$, strict calibration, $10$ seeds, $80$ cases/seed. ``Attack on cnf.'' indicates whether the attacker perturbs the confirmation window inside the same budget. Subject-cluster bootstrap CI uses 60 EEGMMI subjects as resampling units; trial-level Wilson CI is shown for comparison. The B$_{\mathrm{xr}}$ row uses a different EEGMMI run within the same recording day as the confirmation window and is therefore a cross-run within-session audit, not a cross-day deployment estimate.}
\label{tab:phase5-archetype}
\scriptsize
\setlength{\tabcolsep}{2pt}
\resizebox{\linewidth}{!}{
\begin{tabular}{@{}llcccc@{}}
\toprule
Archetype & Pair source & Attack on cnf. & FAR & Wilson CI & Subj.\ CI \\
\midrule
A & same subj.\ same session & no & .053 & $[.039,.070]$ & $[.037,.069]$ \\
B & disjoint subject & no & .048 & $[.035,.065]$ & $[.034,.062]$ \\
B$_{\mathrm{xr}}$ cross-run & same subj.\ disjoint run & no & .040 & $[.029,.056]$ & $[.027,.055]$ \\
C & same subj.\ same session & yes & .998 & $[.991,.999]$ & $[.994,1.000]$ \\
C$'$ & disjoint subject & yes & .999 & $[.993,1.000]$ & $[.996,1.000]$ \\
\bottomrule
\end{tabular}}
\end{table}

\paragraph{Preprocessing stability.}
Repeating the strict $\alpha{=}.05$ frontier under no-record-normalization confirmation acquisition gives FAR $0.046$, CI $[0.034,0.063]$, within sampling noise of the record-normalized $0.053$. The Exp10 preprocessing failure boundary likewise persists under both record-normalized and no-record-normalized strict calibration: smooth9$+$band $0.179/0.240$ and shift-$\pm4$ $0.21/0.23$ across the two pipelines.

\paragraph{Reading.}
Strict split-conformal calibration converts the Exp9 confirmation proxy into a finite-sample frontier whose FAR matches Proposition~\ref{thm:conformal-far} and whose two threat assumptions are individually necessary. The cross-run within-session B$_{\mathrm{xr}}$ archetype shows a pairing that does not require disjoint-subject confirmation in audit, but it is not a cross-day deployment estimate; the C/C$'$ archetypes mark the assumption boundary that any deployment must enforce by acquisition design rather than by post-hoc detection. Cross-day non-stationarity, headset reseating, and impedance drift are open assumptions outside the present data and are revisited in \S\ref{sec:discussion}.

\paragraph{Calibration-size sweep.}
Sub-sampling the calibration negatives to $n_c\in\{160,320,480,640,800\}$ keeps empirical FAR below $\alpha{=}.05$ at all sizes; CI hi shrinks from $0.088$ to $0.070$ as $n_c$ grows, with the $O(1/m_t)$ discretization slack expected from the split-conformal proof. This is the substantive finite-sample check on Proposition~\ref{thm:conformal-far}.

\paragraph{$\varepsilon$-invariance under source-only attack (sanity check).}
Under threat assumption (ii), source perturbations cannot change the confirmation-window score distribution because the confirmer never sees $X_{\mathrm{src}}$. As an audit-pipeline sanity check we sweep $\varepsilon_{\mathrm{src}}\in\{0.04,0.08,0.12,0.16,0.24,0.32\}$ at $\alpha{=}.05$, strict calibration, $10$ seeds, $80$ cases/seed: empirical FAR is exactly $0.0525$ with Wilson 95\% CI upper bound $0.0702$ at every $\varepsilon_{\mathrm{src}}$. The flat sweep is therefore a tautological consequence of the threat model rather than an independent prediction confirmed by data; we report it only to verify that the experimental pipeline does not silently violate (ii), for example by leaking attacker perturbations into the confirmation channel.

\paragraph{Capacity-saturation finding.}
Proposition~\ref{thm:conformal-far} controls FAR for any conformal confirmer regardless of capacity, but does not predict utility. We sweep four EEGNetV4 capacity points by varying $(F_1, D, F_2)$ from $(4,1,4)$ to $(16,2,32)$, an $8\times$ parameter range, holding all other choices fixed at the strict $\alpha=.05$, $\epsilon_{\mathrm{src}}=0.08$ configuration. Table~\ref{tab:phase5-capacity} and Figure~\ref{fig:capacity-saturation} report the result: confirmer accuracy and clean utility are essentially flat across the four sizes, and FAR remains near or below $\alpha$ in every case. The largest model has slightly lower utility ($0.359$ vs.\ $0.385$ for S) and accuracy ($0.806$ vs.\ $0.817$), consistent with mild overfitting on a $\sim 480$-event train split. Within this regime --- EEGMMI R03/R04/R07/R08/R11/R12 left/right command-control with $\sim 480$ training events and EEGNetV4-family backbones --- additional confirmer capacity beyond a few hundred parameters does not improve the FAR/utility trade-off. We do not claim an information-theoretic ceiling: a larger backbone (e.g., EEGConformer or a foundation model linear head) or a larger train index could shift the frontier, and the L row's mild drop is consistent with sample-size-limited overfitting rather than capacity saturation. The reported claim is therefore the within-regime statement that small EEGNetV4 confirmers already saturate the achievable trade-off on this dataset.

\begin{figure}[t]
\centering
\includegraphics[width=\linewidth]{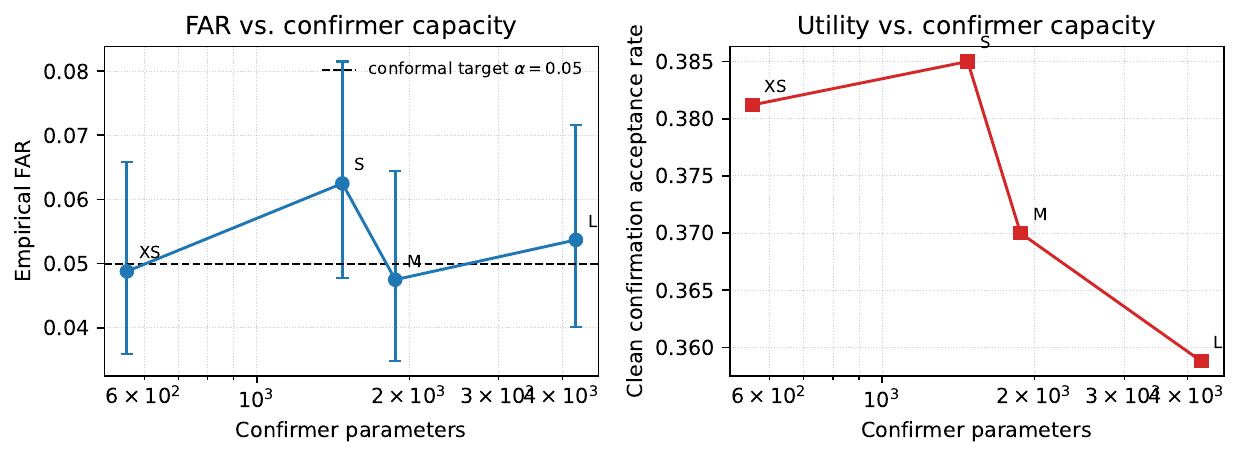}
\caption{Capacity-saturation finding for the conformal confirmer at $\alpha{=}.05$, $\epsilon_{\mathrm{src}}{=}0.08$, $10$ seeds. \textbf{Left:} empirical FAR with $95\%$ Wilson CI as a function of confirmer parameters; the conformal target $\alpha$ is dashed. \textbf{Right:} clean confirmation acceptance rate (utility). FAR and utility are flat across an $8\times$ EEGNetV4 capacity range, indicating dataset-limited rather than confirmer-limited frontier behavior on EEGMMI native command-control.}
\label{fig:capacity-saturation}
\end{figure}

\begin{table}[t]
\centering
\caption{EEGNetV4 capacity sweep at $\alpha{=}.05$, strict calibration, $10$ seeds, $80$ cases/seed. Capacity is varied via $(F_1, D, F_2)$; param counts are computed for the $16$-channel $512$-sample event window. The conformal target line is included for reference.}
\label{tab:phase5-capacity}
\scriptsize
\setlength{\tabcolsep}{2pt}
\begin{tabular}{@{}lccccccc@{}}
\toprule
Cap. & $(F_1,D,F_2)$ & params & conf.\ acc. & FAR & Wilson CI & utility \\
\midrule
XS & $(4,1,4)$ & $554$ & $.816$ & $.049$ & $[.036,.066]$ & $.381$ \\
S & $(8,2,8)$ & $1{,}474$ & $.817$ & $.063$ & $[.048,.082]$ & $.385$ \\
M & $(8,2,16)$ & $1{,}874$ & $.814$ & $.048$ & $[.035,.065]$ & $.370$ \\
L & $(16,2,32)$ & $4{,}258$ & $.806$ & $.054$ & $[.040,.072]$ & $.359$ \\
\bottomrule
\end{tabular}
\end{table}

\paragraph{Failure boundaries and non-evidence cells.}
The validation protocol rejects cells whose clean decoders are not usable or whose target base rate is degenerate. In an earlier log-power LDA defense attempt, the mean clean accuracy was only $0.531$ and zero seeds passed the predefined gate, so that cell is not used as evidence for agreement bypass. A separate target sweep similarly exposed unreachable or degenerate targets in the synthetic harness. These negative checks matter because otherwise C3 bypass could be inflated by weak clean decoders or trivial target priors; the real-EEG A:C claim is restricted to split/seed cells that pass the decoder and denominator gates reported above.

\paragraph{External-consistency bridge.}
The optional live bridge checks whether the deterministic routing semantics survive an actual LLM call under the same fixed prompt policy. It sends only decoded labels, synthetic context strings, provenance and confirmation policy flags, and harmless mock tool schemas; it does not send raw EEG, embeddings, subject identifiers, dataset paths, local file paths, checkpoints, or real tool credentials. In the full public Claude Opus 4.7 run, 144{,}000 decoded-label routing episodes produce parse-valid responses $1.000$, policy-compliant responses $1.000$, deterministic-reference agreement $0.9754$, C2 high-impact route rates $1.000/0.000$ without/with provenance, C3 high-impact route rates $1.000/0.000$ under agreement-plus-provenance versus confirmation-plus-provenance, and clean non-abstain utility $0.846$. The run records 1{,}200 sampled transcripts, 480 live API calls, zero executed high-impact stubs, and zero safety-audit violations. It passes the predeclared 80-case-per-seed decoder-case gate and all router behavior gates. We therefore report the live run as a public transcript bridge, not as a deployed-agent validation or real-tool attack.

\section{Discussion}
\label{sec:discussion}

The experiments support a boundary claim for offline BCI tool-use audits. EEG-only rules cannot see context-only C2 attacks, and dual-channel agreement can be satisfied by a C3 audit intervention that optimizes against both decoders. This does not make agreement useless. Disagreement logs are still useful for burden measurement, attack triage, and abstention policies. Agreement should be evaluated as a monitored component, not presented as a certificate of intent. The result is strongest as an audit-identifiability protocol: each class requires a different observable, a different denominator, and a different falsification cell before a deployment can claim that a route is safe.

\paragraph{Acquisition-isolation regimes and cross-day deployment.}
The threat-archetype matrix in Table~\ref{tab:phase5-archetype} reports three audit regimes for acquisition isolation: same-subject same-session (A), disjoint-subject (B), and same-subject cross-run within session (B$_{\mathrm{xr}}$). All three are within-session audit pairings on EEGMMI R03/R04/R07/R08/R11/R12 recordings collected on the same day. They do not estimate cross-day or cross-session deployment because the dataset does not contain such pairings: there is no time gap, no headset reseating, and no impedance drift between source and confirmation windows in any row. Cross-day non-stationarity has been documented to reduce within-subject EEG transfer, and a deployed cross-day confirmation channel would require a separate field study rather than a stricter audit on this dataset. We therefore label B$_{\mathrm{xr}}$ as a cross-run within-session result, not a deployable cross-day estimate.

\paragraph{Timing alignment is a fragile attack surface, not a robustness margin.}
Exp10 reports that combined smoothing and 8--30~Hz band-limiting collapses agreement-route rates to $0.156$, and that $\pm 4$- and $\pm 8$-sample shifts of an already-optimized perturbation reduce dual-decoder agreement to $0.014$--$0.25$ depending on the constraint family (Table~\ref{tab:exp10-survival}). We treat these as the strongest sanity checks against tensor-only artifacts in the C3 channel. Their reading is asymmetric: they show that the agreement-bypass channel reported by Exp1/Exp6/Exp8 is an offline tensor-aligned audit signal whose persistence under timing perturbations and combined preprocessing is not given for free, not that physical realizability is established. A deployed adversary would need not only a small $\ell_\infty$ budget but also tight timing alignment with the source-window epoch; the present data do not establish either side of physical feasibility, and the temporal-shift sensitivity is a genuine open caveat for any deployment claim built on top of these audits.

\paragraph{Command semantics and deployment boundary.}
The EEGMMI T1/T2 labels are real motor-imagery event labels. In the main R03 route cells, they serve as binary command labels assigned to harmless non-executing route stubs, so those cells should not be read as financial, medical, or file-system command evidence. Exp8 addresses the stronger semantic objection directly by using EEGMMI runs whose T1/T2 annotations correspond to native left/right fist execution or imagery and mapping them only to cursor-left/cursor-right stubs. Exp9 adds an independent EEG confirmation-window proxy, but it remains a proxy for confirmation rather than measured human consent. The broader command-taxonomy cells then test whether the same denominator discipline survives a larger public route vocabulary. The contribution is to make the authorization semantics explicit: a decoded command, an external context, a provenance rule, an agreement rule, and a confirmation policy can be audited separately.

The live bridge is deliberately non-load-bearing. Its purpose is to check whether the decoded-label and context-provenance semantics survive an actual LLM call, not to estimate online BCI-agent risk. Claude Opus 4.7 reproduces the C2/C3/confirmation route pattern over 144{,}000 routing episodes with parse validity and policy compliance both $1.000$, and the run passes the full 80-case-per-seed gate. The public tool-agent and command-taxonomy transcript cells remain the main route-semantics evidence because they use decoded real-EEG command labels, public synthetic contexts, and the same 10-seed protocol without sending private neural data. A deployed-agent study would fix a public tool-calling deployment, publish transcripts, add user-confirmation labels, and evaluate many more command contexts. That study would test external validity rather than replace the C2/C3 observability claim.

The C3 lemma is not a surprise claim about optimization. Its role is to identify a predictive feasibility condition: if a shared raw-input perturbation can cross each decoder's target margin, agreement becomes an attacker-side constraint. The empirical contribution is the falsification test that reviewers can inspect: independently trained raw EEG decoders, three split protocols, a perturbation-strength sweep, explicit margin diagnostics, clean-correct agreement subset auditing, public tool-routing semantics, transfer attacks, smoothed / band-limited / causal-ramp constraints, post-hoc preprocessing and temporal-shift stress, multi-step adversarial verification, matched clean-burden calibration, ensembles, score-query attacks, and decision-only route/abstain searches all show where the shared feasible target region persists and where it breaks under the stated offline audit condition. Stronger non-confirmation verification reduces route rates, but it does not convert agreement into an intent certificate.

Finally, this is not an online BCI deployment study, a physical-injection demonstration, a human-consent study, or a real-tool authorization benchmark. Hardware latency, nonstationarity, user confirmation behavior, richer command semantics, and clinical risk controls are outside the present evidence. The standardized raw perturbation budget is calibrated against the event cache, and the route/stress cells add smoothed, band-limited, causal-ramp, temporal-shift, transfer, matched-defense, score-query, decision-only, native-command, and independent-confirmation-proxy constraints, but they are still offline tensor-space audits rather than actuator models. Those factors are the next validation layer; the top-level contribution here is that BCI tool-use pipelines need class-specific validation because their failure modes live in different observable channels.

\paragraph{Next validation gates.}
The present evidence establishes a class-specific audit failure in calibrated offline cells. The next strengthening step is not to weaken the claim, but to add external-validity layers after the audit variables are identifiable: real user-confirmation labels, public online tool-calling transcripts, temporal/nonstationarity measurements, and physical-channel or human-in-the-loop perturbation constraints. These extensions would test deployment validity while preserving the core contribution: each route-safety claim requires a distinct observable and falsification denominator.

\paragraph{C4 query logging.}
C4 is not a main empirical claim in this paper. It is kept as a reporting row because query access changes what an auditor can infer about decoder, embedding, or policy behavior. Without a direct extraction experiment, the requirement is limited to disclosure: BCI-tool papers should state query access, logging, and rate limits rather than treating decoder behavior as unobservable.

\paragraph{Ethics and dual-use boundary.}
The experiments use public EEGMMI-style event labels, harmless non-executing routes, deterministic stubs, and optional prompts that omit raw EEG, embeddings, subject identifiers, and dataset paths. The attack code is useful for defensive evaluation because it exposes when agreement should not be treated as an intent certificate. A release should therefore package the protocol, metrics, and harmless surrogate routes, while avoiding turnkey instructions for operating against deployed BCI devices or real tool actions.

\section{Conclusion}
\label{sec:conclusion}

BCI-controlled tool-use pipelines require route-safety audits whose logged variables match the attack class. The minimal audit-schema theorem, C3 attacked-dependence decomposition, fixed routing suite, native command-control endpoint, non-oracle confirmation-proxy frontier, temporal/preprocessing stress frontier, and real EEG validation protocol show why EEG-only rules and clean agreement cannot certify C2/C3 authorization claims by themselves. The practical implication is not to discard verification, but to evaluate it as one component in a class-specific safety stack with provenance, confirmation, and adaptive dual-decoder testing under explicit route and perturbation constraints.

\begingroup
\sloppy

\endgroup

\end{document}